\documentclass[11pt]{article}

\usepackage{color}
\definecolor{boxshade}{gray}{0.85}
\usepackage{latexsym}
\usepackage{ifthen}
\usepackage{verbatim}
\ifx\documentclass\undefined
\else
  \DeclareSymbolFont{tlaitalics}{\encodingdefault}{cmr}{m}{it}
  \let\itfam\symtlaitalics
\fi

\usepackage{pdfpages} 

\usepackage{tikz}
\usepackage{amsmath}
\usepackage{amsthm}
\usepackage{amssymb}
\usepackage{hyperref}

\newtheorem{definition}{Definition}
\newtheorem{proposition}{Proposition}
\newtheorem{corollary}{Corollary}

\usetikzlibrary{arrows}
\usetikzlibrary{shapes}
\usetikzlibrary{calc}
\usetikzlibrary{patterns}
\usetikzlibrary{fit}
\usetikzlibrary{positioning}
\usetikzlibrary{decorations.pathreplacing}

\gdef\dash---{\thinspace---\hskip.16667em\relax}
\gdef\smdash--{\thinspace--\hskip.16667em\relax}
\gdef\op|{\,|\;}

\usepackage{enumitem} 

\usepackage{listings, listings-rust}
\lstdefinelanguage{Go}{
  morekeywords=[1]{
    break,default,func,interface,select,case,defer,go,map,
    struct,chan,else,goto,package,switch,const,fallthrough,
    if,range,type, continue,for,import,return,var},
  morekeywords=[2]{
    append,cap,close,complex,copy,delete,imag,
    len,make,new,panic,print,println,real,recover},
  morekeywords=[3]{
    bool,byte,complex64,complex128,error,float32,float64,
    int,int8,int16,int32,int64,rune,string,
    uint,uint8,uint16,uint32,uint64,uintptr},
  morekeywords=[4]{true,false,iota,nil},
  morestring=[b]{"},
  morestring=[b]{'},
  morestring=[b]{`},
  comment=[l]{//},
  morecomment=[s]{/*}{*/},
  sensitive=true
}
\lstdefinelanguage{tla}
  {morekeywords={MODULE,EXTENDS,CONSTANTS,CONSTANT,ASSUME,VARIABLES,VARIABLE,
          EXCEPT,UNCHANGED,TRUE,FALSE,IF,THEN,ELSE,LET,IN,SUBSET},
          comment=[l]{\\*},
  morecomment=[s]{(*}{*)},
  mathescape=true,escapechar={@},
  basicstyle=\small\sffamily,
  commentstyle=\itshape\rmfamily\small,
  keywordstyle=\sc
}

\definecolor{dkgreen}{rgb}{0,0.6,0}
\definecolor{gray}{rgb}{0.5,0.5,0.5}
\definecolor{mauve}{rgb}{0.58,0,0.82}

\newcommand{\tlap}[0]{\textsc{TLA}\textsuperscript{+}}

\newcommand{\primary}{\mathit{primary}}
\newcommand{\bheight}{\mathit{Height}}
\newcommand{\btime}{\mathit{Time}}
\newcommand{\vals}{\mathit{Val}}
\newcommand{\nvals}{\mathit{NextV}}
\newcommand{\commit}{\mathit{LastCommit}}
\newcommand{\blockid}{\mathit{LastBlockID}}
\newcommand{\vtype}{\mathit{Type}}
\newcommand{\vround}{\mathit{Round}}
\newcommand{\msgblockid}{\mathit{BlockID}}
\newcommand{\votes}{\mathit{Votes}}
\newcommand{\fhash}{\mathit{hash}}
\newcommand{\targeth}{\mathit{targetHeight}}
\newcommand{\lightStore}{\mathit{lightStore}}

\newcommand{\fexecute}{\mathit{execute}}
\newcommand{\fproof}{\mathit{proof}}
\newcommand{\fmatchproof}{\mathit{matchProof}}
\newcommand{\fmatchhash}{\mathit{matchHash}}
\newcommand{\PossibleCommit}{\mathit{PossibleCommit}}
\newcommand{\chain}{\mathsf{chain}}
\newcommand{\totalvp}{\mathsf{totalVP}}
\newcommand{\correct}{\mathsf{correctUntil}}
\newcommand{\trustp}{\mathit{TP}}
\newcommand{\current}{\mathsf{current}}
\newcommand{\nextHeight}{\mathsf{nextHeight}}
\newcommand{\LatestVerified}{\mathsf{LatestVerified}}

\newcommand{\startHeader}{\mathit{sh}}
\newcommand{\startTime}{\mathit{sTime}}
\newcommand{\trustedHeader}{\mathit{trustedHeader}}

\DeclareMathOperator{\UNION}{\textsc{union}}

\DeclareMathOperator{\SUBSET}{\textsc{subset}}                                  
\DeclareMathOperator{\DOMAIN}{\textsc{domain}}

\newcommand{\A}{\forall}
\newcommand{\E}{\exists}
\newcommand{\tladef}[0]{\triangleq{}}
\newcommand{\TRUE}[0]{\textsc{true}}
\newcommand{\FALSE}[0]{\textsc{false}}
\newcommand{\BOOLEAN}[0]{\textsc{boolean}}

\usepackage{pifont}
\newcommand{\cmark}{\ding{51}}
\newcommand{\xmark}{\ding{55}}
\newcommand{\bcheck}[1]{$\text{\cmark}_{\le#1}$}
\newcommand{\bmark}[1]{$\text{\xmark}_{=#1}$}
\newcommand\tbh[1]{\textsf{\textbf{\scriptsize{#1}}}}

\AtBeginDocument{
  \providecommand\BibTeX{{
    \normalfont B\kern-0.5em{\scshape i\kern-0.25em b}\kern-0.8em\TeX}}}

\usepackage{authblk}

\title{A Tendermint Light Client\thanks{Supported by Interchain Foundation}}

\author[1]{Sean Braithwaite}
\author[1]{Ethan Buchman}
\author[1,2]{Ismail Khoffi}
\author[1]{Igor Konnov}
\author[1]{Zarko Milosevic}
\author[1]{Romain Ruetschi}
\author[1]{Josef Widder}
\affil[1]{Informal Systems}
\affil[2]{LazyLedger}
\begin{document}

\maketitle

\begin{abstract}
In Tendermint blockchains, the proof-of-stake mechanism and the
     underlying consensus algorithm entail a dynamic fault model that
     implies that the active validators (nodes that sign blocks) may
     change over time, and a quorum of these validators is assumed to
     be correct only for a limited period of time (called trusting
     period).
The changes of the validator set  are under control of the blockchain
     application, and are committed in every block.
In order to check what is the state of the blockchain application at
     some height $h$, one needs to know the validator set at
     that height so that one can verify the corresponding digital signatures
     and hashes.
A na\"\i{}ve way of determining the validator set for height $h$
     requires one to:  (i) download all blocks before $h$, (ii) verify
     blocks by checking digital signatures and hashes and (iii)
     execute the corresponding transactions so the changes in the
     validator sets are reproduced.
This can potentially be very slow and computationally and data
     intensive.
     
In this paper we formalize the dynamic fault model imposed by
     Tendermint, and describe a light client protocol that allows to
     check the state of the blockchain application that, in realistic
     settings, reduces significantly the amount of data needed to be downloaded,
     and  the number of required computationally expensive signature
     verification operations.
In addition to mathematical proofs, we have formalized the light
     client protocol in \tlap{}, and checked safety and liveness with
     the APALACHE model checker.
\end{abstract}

\newcommand{\drawblock}[9]{
  \node (height#1) at (#2 cm, #3) [typetag, anchor=west, xshift=1mm]
        { $\Height \mapsto #4$};
   \node (time#1) [typetag, below=of height#1.west, anchor=west]
         { $\Time_{#4}$};

     \node (data#1) [typetag, below=of time#1.west, anchor=west]
        { Data, AppState, ...};
  \node (val#1) [typetag, below=of data#1.west, anchor=west] { $\Val_{#4}\mapsto #5$};        
  \node (nv#1) [typetag, below=of val#1.west, anchor=west] { $\NVal_{#4} \mapsto #6$};

  \node (lcsig#1) [typetag, below=1.5cm of nv#1.west, anchor=west,
    xshift=2mm,minimum width=3.4cm,text width=3.4cm] { $\Votes_{#4} \mapsto #8$};

  \node (lcbID#1) [typetag, below=of lcsig#1.west, anchor=west,
    minimum width=3.4cm,text width=3.4cm] { $\BlockID_{#9}$};
  \node[draw=black,rounded corners=3pt, fit=(lcsig#1)(lcbID#1),
        label={[fill=gray!90,font=\color{white}\scriptsize\ttfamily]above:$\Commit_{#4}$: LastCommit }
  ] (lcbox#1) {};

  \node[draw, thick, #7rounded corners=4pt, fit=(height#1)(lcbox#1), 
        label={[title,rotate=90,xshift=.25cm]above left:Block #4}] (block#1) {};
}

\section{Introduction}
\label{sec:intro}

Tendermint is a leading state machine replication (SMR)
     engine~\cite{Buchman2016} that tolerates Byzantine faults. It
     supports arbitrary state machines written in any programming
     language by providing a flexible interface for application development.
Its core consensus protocol~\cite{bkm2018latest} is a variant of the
     algorithm for Byzantine faults
with Authentication from~\cite{DLS88}, built on top of an
     efficient gossiping layer.

Tendermint is designed for public, open-membership networks, and a
     production deployment  consists of a number of node types.
The consensus-forming nodes, responsible for proposing blocks and
     voting on them, are called \emph{validators}.
Validators are not directly connected to one another, but rather
     through a gossip network consisting of non-validating nodes
     called \emph{full nodes}, who execute the same consensus protocol
     and relay its messages but do not participate in block production
     or voting.
Besides this restriction, a full node also maintains the copy of the
     blockchain,  and executes the same set of protocol rules like
     validators.
Since validator nodes are highly security-sensitive (as they need to
     manage private keys), their operators often run additional full
     nodes, termed \emph{sentry nodes}, as gateways to the rest of the
     network.
A validator node typically connects only to its sentry nodes, which
     then form connections to other full nodes and sentry nodes.

A final type of node is a light node, also called a \emph{light
     client}.
Light nodes do not participate formally in the network; they make
     requests for data from full nodes and verify, by checking hashes
     and signatures, that such data indeed came from the underlying
     blockchain.
Light clients synchronize to the latest block of the blockchain by
     verifying validator signatures.
Once synchronized, they can verify  Merkle proofs about the state of
     the blockchain by using the Merkle root stored in the latest
     block.
Effectively, light nodes perform read operations of the blockchain.
However, since light nodes do not follow the full consensus protocol
     or execute transactions, they operate under a different security
     model than full nodes.
A light client protocol is normally designed  such that it does not
     have high computational and bandwidth requirements as it is
     supposed to be used also in constrained environments, for example
     in mobile devices.

In a traditional Byzantine-fault-tolerant SMR system, a read operation
     by a client is normally implemented by having the client send a
     request to all replicas and then waiting to receive the same
     response from at least $f+1$ replicas, where $f$ is the maximum
     number of faulty replicas.
This approach is not applicable in the Tendermint model as (i) a
     client does not have direct access to  validator nodes (as they
     are connected to the network only through sentry nodes that are
     not necessarily public)  and (ii) validator set changes are
     dynamic so it is not obvious what is the current validator set.
In Tendermint, no restrictions are placed on how the validator set may
     change from one block to the next; the intersection between
     adjacent validator sets may be empty.
Changes to the validator set are determined entirely by the
     application state machine.

In this paper we describe a protocol that addresses the challenges (i)
     and (ii) for the design of a light client.
As the light client cannot query a validator directly, it cannot query
     the required $f+1$ validators needed for the traditional
     approach.
The described shielding of validator nodes by sentries thus forces us
     to design a protocol that allows the light node to (a) obtain the
     required data from other nodes (different from validators), and
     (b) verify the received information by checking signatures (in
     contrast to checking $f+1$ identical responses).
In addition, in a context where the validator set is always changing,
     it is not clear to a light client which validators can  currently
     be trusted, and who has the authority to sign a block as a
     validator.
The essential operation of a light node thus becomes tracking the
     evolving validator set over time.

             \newcommand{\NVal}{\mathrm{NextV}}
\newcommand{\Val}{\mathrm{Val}}
\newcommand{\Votes}{\mathrm{Votes}}
\newcommand{\BlockID}{\mathrm{BlockID}}
\newcommand{\Commit}{\mathrm{Commit}}
\newcommand{\Height}{\mathrm{Height}}
\newcommand{\Time}{\mathrm{Time}}

\begin{figure*}[t]
\begin{tikzpicture}[
  >=latex,
  node distance=5mm,
  title/.style={fill=gray!90,font=\fontsize{8}{8}\color{white}\ttfamily},
  typetag/.style={rectangle, draw=gray!50, font=\scriptsize\ttfamily,
    minimum width=3.8cm, text width=3.5cm}
]

\drawblock{b1}{0}{0}{1}{p_1,p_2,p_3,p_4}{p_1,p_2,p_3,p_4}{}{\emptyset}{0}
\drawblock{b2}{5.15}{0}{2}{p_1,p_2,p_3,p_4}{p_3,p_4,p_5,p_6}{}{p_1,p_2,p_3}{1}
\drawblock{b3}{11.5}{0}{3}{p_3,p_4,p_5,p_6}{p_3,p_4,p_5,p_{42}}{}{p_2,p_3,p_4}{2}

\draw [decorate,decoration={brace,amplitude=10pt},xshift=4pt,yshift=0pt]
($(blockb2.north west) + (4.3,0)$) -- ($(blockb2.south west) + (4.3,0)$) node (b2p)
      [black,midway,xshift=0.6cm] {};

\draw [thick,->]  (lcbIDb3.west) to[bend left=10]
    node[anchor=north east, sloped,near start]
    {\textbf{\textsf{hash}}} (b2p.west);

\draw[thick,->] (nvb1.east) to[bend left=10]
    node[anchor=north, sloped] {\textbf{\textsf{=}}} (valb2.west);

\end{tikzpicture}

\caption{\boldmath{Sequential Verification. If Block 1 is trusted to be from the
blockchain, then $\NVal_1$ defines the validators $\Val_2$ that vote for Block
2. If a commit for Block 2 (recorded as the LastCommit
in Block~$3$), contains votes from 
more than $2/3$ of the voting power in $\NVal_1=\Val_2$, 
then Block 2 can be trusted to be from the blockchain.}}
\label{fig:sequ-verif}
\end{figure*}
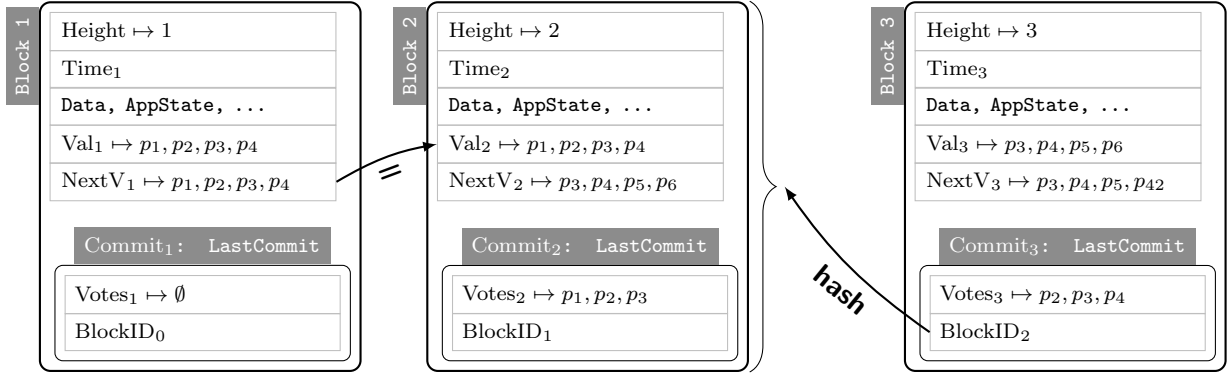

Tendermint is particularly popular for proof-of-stake blockchains.
In such blockchains, becoming a validator requires an economic
     commitment, often referred to as a bond (or stake), which
     guarantees the correct behaviour of the validator.
Validators that behave incorrectly (i.e., by  going offline, or
     signing conflicting blocks) lose some fraction of their bond.
Thus, bonded validators have an incentive to behave correctly.
For a validator to reclaim its stake, it must wait for a so-called
     \emph{unbonding period} before their stake is returned.
This unbonding period is sufficiently long to detect misbehavior and
     punish it.
But this also means that participants only have the incentive to
     follow the protocol for some limited time.
Once they can be sure that they will get their stake back, it is
     unsafe to rely on their cooperation.
These Proof-of-Stake mechanics result in a Tendermint Security Model
     with time-dependent fault assumptions that depend on the blocks.

Tendermint is also a core component of the Cosmos
     Project~\cite{cosmos}, which consists of many independent
     proof-of-stake blockchains.
At the heart of the Cosmos Project is the InterBlockchain
     Communication  (IBC) protocol for reliable communication between
     independent blockchains; what TCP is for computers, IBC aims to
     be for blockchains.
IBC is  based on light client protocols, which enable one
     blockchain to perform read operations on another.

\paragraph{Contributions.}
The major contribution of this paper is to
provide the formal underpinning for Tendermint light
clients, which rests on three pillars
\begin{itemize}
\item A formalization of the Tendermint security model.
\item A light client protocol based on the security model, and modeled in TLA+.
\item Model Checking the correctness of the protocol against the
  security model with the Apalache model checker.
\end{itemize}

We start by formalizing the Tendermint Security model; we then discuss
     the design of the Light Client Verification that implements a
     fault-tolerant read based on the model.
The formalization of the failure model,  and the invariants introduced
     in this paper are based on the open source reference
     implementation~\cite{TMCORE} and the specifications of the data
     structures~\cite{TMBC}.
The challenge addressed here is that the light client might have a
     block of height $h_1$ and needs to read the block of height $h_2
     > h_1$.
Checking all block headers of heights from $h_1$ to $h_2$ might be too
     costly (e.g., in terms of energy for mobile devices).
The described protocol tries to reduce the number of intermediate blocks that
     need to be checked, by exploiting the guarantees provided by the
     security model.

In the following we outline our approach and
     start by describing the Tendermint block structure and how nodes
     follow the state of the application in the standard way.

\section{Overview}
\label{sec:glance}

\paragraph{Tendermint Signature Scheme}

Figure~\ref{fig:sequ-verif} shows an example of the first three blocks
     generated by a Tendermint blockchain.
Block~1 is the result of an instance of Tendermint consensus.
The instances are called heights.
The validators who actively participated in this consensus instance and 
their IDs (public key) are stored in $\vals_1$.
In Tendermint, instead of having one vote each, validators may have
     different voting powers, which are also stored in $\vals_1$.
We say that in consensus a quorum is reached if validators that
     represent more than $2/3$ of the voting power agree.

In addition to deciding upon data to be put into a
     block, these validators also decide on the nodes who are going to
     participate in the next height, and the IDs (public keys) of
     these nodes are stored in $\nvals_1$.
The same nodes are stored in $\vals_2$, and they decide on Block~2.
The proof that they indeed decided on Block~2 is given by a commit.
A commit contains a set of signed messages that contain a hash called
     $\msgblockid_2$, and in the blockchain it is stored in Block~3.
In order for a commit to be valid, it must contain messages by a
     quorum of validators.
In our example, we assume each validator has voting power 1, and
     $\votes_3$ contains signatures by more than $2/3$ of the nodes in
     $\vals_2$.
Thus, if a node has obtained Block~1 from a trusted source, and it has
     downloaded Block~2 and Block~3, it can check that Block~2 is
     indeed from the blockchain, by checking the described
     relationships of hashes, quorums, and signatures.

     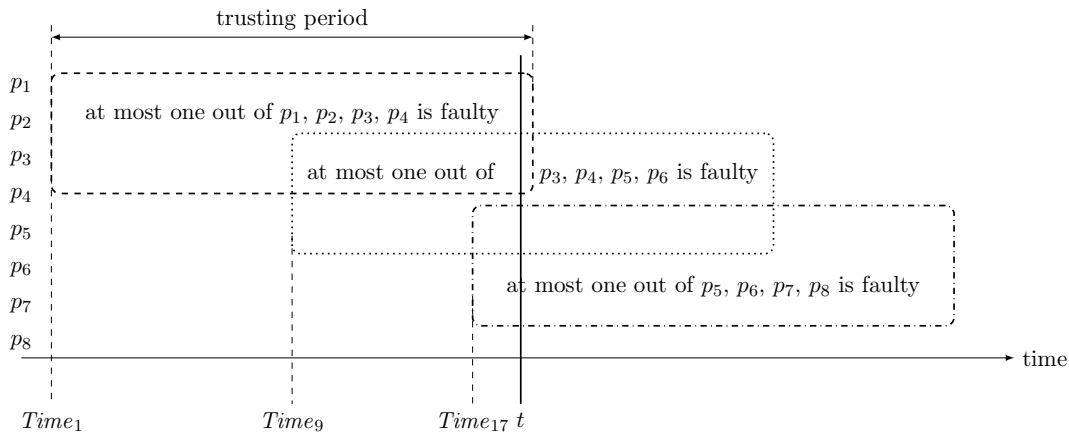
\begin{figure*}[t]
  \scalebox{0.8}{
\begin{tikzpicture}[
  >=latex,
  node distance=1mm,
  title/.style={font=\fontsize{6}{6}\color{black!50}\ttfamily},
  typetag/.style={rectangle, draw=black!50, font=\scriptsize\ttfamily,
    minimum width=2.8cm, text width=2.5cm}
]

  \node (p1) at (0,0) {$p_1$};
  \node (p2) [below=of p1] {$p_2$};
  \node (p3) [below=of p2] {$p_3$};
  \node (p4) [below=of p3] {$p_4$};
  \node (p5) [below=of p4] {$p_5$};
  \node (p6) [below=of p5] {$p_6$};
  \node (p7) [below=of p6] {$p_7$};
  \node (p8) [below=of p7] {$p_8$};
 \node (end) [right=16.5cm of p8.south] {time};
 \draw [->] (p8.south) --  (end);

 \draw [dashed] (0.5,1) -- (0.5,-5.3) node [anchor=north] {$\btime_1$};
 \draw [dashed] (8.5,1) -- (8.5, -0.1);
 \draw [<->] (0.5,0.8) -- node[anchor=south] {trusting period}  (8.5,0.8);
  \draw [dashed] (4.5,-2.5) -- (4.5,-5.3) node [anchor=north] {$\btime_9$};
  \draw [dashed] (7.5,-3.5) -- (7.5,-5.3) node [anchor=north] {$\btime_{17}$};
  \draw [thick] (8.3,0.5) -- (8.3,-5.3) node [anchor=north] {$t$};

\node[fit={(0.5,0.2) (8.5,-1.8)}, inner sep=0pt, draw, dashed,
  thick,rounded corners=4pt, thick]
(rect) {at most one out of $p_1$, $p_2$, $p_3$, $p_4$ is faulty\\~};
\node[fit={(4.5,-0.8) (12.5,-2.8)}, inner sep=0pt, draw, dotted,
  thick,rounded corners=4pt, thick]
(rect) {at most one out of \hspace{.5cm} $p_3$, $p_4$, $p_5$, $p_6$ is faulty\\~};
\node[fit={(7.5,-2.0) (15.5,-4.0)}, inner sep=0pt, draw, dash dot,
  thick,rounded corners=4pt, thick]
(rect) {~\\~\\at most one out of $p_5$, $p_6$, $p_7$, $p_8$ is faulty};

\end{tikzpicture}
}
\caption{\boldmath{Example for the Tendermint security model where all
  nodes have voting power at most 1. At $\btime_1$ a
  block $b_1$ with $\nvals_1 = \{ p_1, p_2, p_3,
  p_4\}$ is generated.
  This imposes a fault assumption during the
  trusting period that at most one of the nodes are faulty as $n>3f$,
  where $n$ is the number of nodes and $f$ is an upper bound on the
  number of faults.
  The rectangles depict the trusting period and which
nodes are bound by the Tendermint Security Model. At the points
outside the rectangles, the nodes are not bound by a fault
assumption.
We do not show the fault assumptions
imposed by blocks other than $b_1$, $b_9$, and~$b_{17}$.}} 
\label{fig:tfm}
\end{figure*}

An important feature of Tendermint is that the choice of $\nvals$ is
     application-specific and unrestricted; for any height~$i$, the sets $\vals_i$ and $\nvals_i$ need
     not intersect.
Thus, if one needs to check some data, e.g., the existence of a
     transaction in a block of height~$\ell$, one a priori
     needs to download the blocks for all heights up to $\ell+1$, and
     sequentially check all the blocks.
This is computationally expensive due to checking hashes and
     signatures, and may hinder the access to blockchain, e.g., for
     mobile devices.
     That said, in many deployments of Tendermint, e.g.,
     on the Cosmos Hub blockchain,
     we observe that large changes in the validator sets are rare.
Thus the question is to design a protocol that allows to verify the
     block of height~$\ell$ without downloading and verifying all the
     blocks required by the sequential method.
We call such a protocol a (skipping) light client.
It implements a read operation of a block, by
     communicating with full nodes.
As some full nodes may be faulty, this functionality must be
     implemented in a fault-tolerant way.
To do so, we next formalize the fault assumption that is imposed by
     Tendermint.

     \paragraph{Tendermint Security Model}

The staking and unbonding mechanism induces a security
     model: starting at the time a block gets generated (this time is
     stored in the block), more than two-thirds of the next validators
     of a new block are correct for the duration of
     the \emph{trusting period}, a duration which is less than the unbonding
     period defined by the Proof-of-Stake mechanics. 
An example is sketched in Figure~\ref{fig:tfm}: Block~1 is created at
     time $\btime_1$, and the full nodes $p_1$, $p_2$, $p_3$, and
     $p_4$ are decided to be the validators of the next block.
More than two thirds, that is, at least three, are assumed to be
     correct during the time interval depicted by the dashed rectangle.
Similarly, when Block~9 is generated, an additional constraint over
     $p_3$, $p_4$, $p_5$, and $p_6$ is added.
This so-called security model can be seen as a Byzantine fault model
     with dynamic (or moving) faults. 
     
The fault-tolerant read operation over a Tendermint blockchain needs
     to be designed for this security model.
To do so, after formalizing the blockchain data structure in
     Section~\ref{sec:blockchain}, we formally define this security
     model in Section~\ref{sec:TMFM}.

It should be noted that Tendermint provides 
     guarantees even outside this security model, where one third or more validators
     are faulty within a trusting period and may thus fork the blockchain.
     The security model introduced here enables clear separation of concerns 
     between light client verification, which operates within the model,
     and fork detection, which operates outside it. We leave fork detection to future work.

\paragraph{Skipping Verification under Tendermint Security Model}

In Section~\ref{sec:problem} we formalize the distributed computing
     problem for which our light client is designed, and prove the
     central result that allows us to solve it in
     Section~\ref{sec:checks}.
The underlying intuition is as follows:  Recall Figure~\ref{fig:tfm}.
Assume a light client starts with Block~1 from the blockchain, and
     needs to verify Block~{17} at time~$t$.
Consider the example in Figure~\ref{fig:skip-verif}.
Here, $\nvals_1 = \{p_1,p_2,p_3,p_4\}$, which by the Tendermint
     security model means that more than two thirds of these nodes are
     correct during the trusting period, which is depicted by the dashed
     rectangle in Figure~\ref{fig:tfm}.
Let's assume the light client has downloaded blocks of height~$17$
     and~$18$ with a validator set  $\vals_{17} =
     \{p_5,p_6,p_7,p_8\}$.
The light client would like to verify Block~{17} based on
     $\commit_{18}$, which would need to contain a quorum of
     $\vals_{17}$.
However,  being aware only of the dashed rectangle, it is impossible to
     infer whether any of the nodes in $\vals_{17}$ are correct.
Thus Block~{17} cannot be verified at this time.

The idea behind the light client is to try bisection, i.e., since
     the first check for  Block~{17} failed, the light client will
     download a header in the middle, i.e., Block~9 from some node in
     the Tendermint network.
To verify Block~9, the light client also needs a commit for that block
     which is stored in the next block, so that it also downloads
     Block~{10}.
Let's consider the node has downloaded Blocks~9 and~{10} as
     shown in Figure~\ref{fig:skip-verif}.
At time~$t$, the node can still trust that less than one third of the
     nodes in $\nvals_1$ is faulty.
The set $\votes_{10}$ contains two of the nodes from $\nvals_1$, namely,
     $p_3$ and~$p_4$.
By the security model, at least one of them is necessarily correct.
As a result, at least one of the nodes who signed Block~9 is correct.

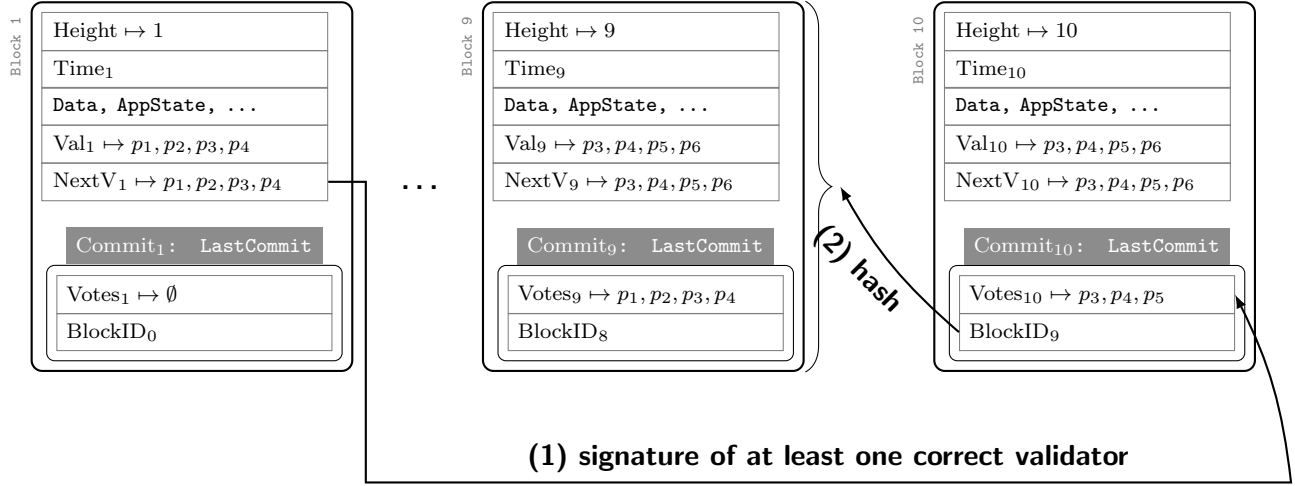
\begin{figure*}[t]
\begin{tikzpicture}[
  >=latex,
  node distance=5mm,
  title/.style={font=\fontsize{6}{6}\color{black!50}\ttfamily},
  typetag/.style={rectangle, draw=black!50, font=\scriptsize\ttfamily,
    minimum width=3.8cm, text width=3.5cm}
]

\drawblock{bi}{-1}{0}{{1}}{p_1,p_2,p_3,p_4}{p_1,p_2,p_3,p_4}{}{\emptyset}{0}
\drawblock{bj}{5}{0}{{9}}{p_3,p_4,p_5,p_6}{p_3,p_4,p_5,p_6}{}{p_1,p_2,p_3,p_4}{8}
\drawblock{bjp1}{11}{0}{{10}}{p_3,p_4,p_5,p_6}{p_3,p_4,p_5,p_6}{}{p_3,p_4,p_5}{9}

\node (dots) [right=of blockbi] {\textbf{\textsf{\dots}}};

\draw [decorate,decoration={brace,amplitude=10pt},xshift=4pt,yshift=0pt]
($(blockbj.north west) + (4.3,0)$) --  ($(blockbj.south west) + (4.3,0)$) node (bjp)
      [black,midway,xshift=0.6cm] {};

\draw [thick,->]  (lcbIDbjp1.west) to[bend left=10]
    node[anchor=north east, sloped,near start]
    {\textbf{\textsf{(2) hash}}} (bjp.west);

\newdimen\mydimi
\pgfextracty{\mydimi}{lcsigbjp1.east}
\draw [thick,->]  (nvbi.east) -- +(0.5,0) -- +(0.5,-4) -- node[anchor=south]
      {\textbf{\textsf{(1) signature of at least one correct validator}}} +(12.8,-4) to[bend right=10]
       (lcsigbjp1.east);

\end{tikzpicture}
\caption{\boldmath{
Example for Skipping Verification.
The light client has downloaded blocks of height 1, 9, 10, and has
     trust that Block~1 was generated by a Tendermint blockchain.
Moreover assume that the trusting period of Block~1 has not expired, that is, the
     current time  $t < \btime_1 + \trustp$.
Consider the two checks (1) and (2): (1)  $\Votes_{10}$ contains one
     correct validator from  $\NVal_1$, that is, it contains
     validators that represent more than $\frac{1}{3}$ of the voting
     power in $\NVal_1$.
(2) The hash matches.
By the behavior of a correct node executing Tendermint consensus, if
(1) and (2) are satisfied, then $\vals_9$ is indeed the validator set of
height 9 on the blockchain, and
by the Tendermint Security Model,
the light client can
     trusts $\Commit_{10}$, and thus Block~$9$ was generated by a
     Tendermint blockchain and can be trusted.}  }
\label{fig:skip-verif}
\end{figure*}
 
In Tendermint consensus, correct nodes only sign blocks that were
     properly generated; thus, Block~9 in the figure can be trusted.
As the light client has now established that Block~9 is from
     the blockchain, $\nvals_9$ from
     Figure~\ref{fig:skip-verif} imposes a trust assumption that
     corresponds to the dotted rectangle in Figure~\ref{fig:tfm}.
Thus, the light client may now try to verify Block~{17} based
     on this new trust assumption.

As the Tendermint security model makes reference to real time, and the
     trusting period is a concrete time duration, we require that the
     light nodes local time is approximately synchronized to the time
     of the Tendermint blockchain.
This is needed to check whether a block is within the trusting period.
However, the trusting period is in the order of weeks, so that in
     typical scenarios a clock precision of several seconds is
     sufficient and easily achievable.
Similarly, for liveness we require that downloading a header is faster
     than the duration of the trusting period (i.e., two weeks).
While from a theoretical viewpoint, all this implies that we operate
     in a synchronous computation model, in practice, due to the order
     of time durations, this does not impose practical limitations.
     
In the implementation we do several performance improvements.
Rather than downloading complete blocks, we download so-called
\emph{lightblocks} that just contain the required metadata to do the checks
     above.
Such a lightblock for Block~9 contains the header of the Block~9 and
     also $\commit_{10}$, as this is the only information we need from
     Block~$10$. Thus, instead of downloading Blocks~9 and~10, the
     implementation downloads the lightblock~9, only.

\paragraph{Model Checking and Implementation}

In Section~\ref{sec:tla} we discuss how we have formalized the
     blockchain and the light client protocol in~\tlap{}.
In addition to having a machine-readable protocol specification
     in~\tlap{} (that reads very similar to the mathematical
     description from Sections~\ref{sec:blockchain}--\ref{sec:proto}),
     we were able to produce non-trivial system executions with the
     symbolic model checker \textsc{Apalache}, and we checked the
     protocol for small instances of the blockchain.
In Section~\ref{sec:impl} we discuss our implementation in Rust that
     is based on a modular architecture that simplifies testing.

\section{Blockchain data structure}
\label{sec:blockchain}

We give a formalization of the Tendermint block structure.
We start with an abstract view, in particular with respect to the
     domains of data fields, to highlight concepts as independent of
     their implementation as possible.
We will later refine towards the implemented data structures~\cite{TMBC} to
     address distributed aspects.

A set of  transactions is stored in a data structure called
     \emph{block}, which contains a field called \emph{header}.
In the implementation, hashes are used to reduce the amount of data
     that needs to be (re)transmitted and stored.
Hashes are used within a block, where a header stores hashes of the
     data of the block.
But hashes are also used to point to the previous block.
The former usage is done just for performance, so that for our
     purposes we ignore these hashes, and we will assume that the
     blockchain is a list of headers, rather than a list of blocks.
The hashes that point to previous blocks are needed to implement that
     chain, and we will treat them explicitly in this model.

\begin{definition}[Header]\label{def:header}
A header contains the following fields, whose domain (except the
height)
we leave
unspecified for now:

\begin{itemize}

\item
  $\bheight$: non-negative integer
\item
  $\btime$
\item
  $\commit$
\item
  $\blockid$
\item
  $\vals$
\item
  $\nvals$
\item
  \texttt{Data}
\item
  \texttt{AppState}
\item
  \texttt{LastResults}
\end{itemize}
\end{definition}

In the implementation, $\blockid$ is also stored as part of $\commit$ as
indicated in Figure~\ref{fig:sequ-verif}, and is a hash of the
previous block. For our theoretical treatment it is more convenient to
not treat it within the $\commit$. This redundancy is also subject to
ongoing discussions in the Tendermint project~\cite{LastBlockID}.

Tendermint consensus~\cite{bkm2018latest} generates a sequence of such
headers, that ensures the following invariants:

\begin{definition}[Basic Invariants] \label{def:basicinv}
  A Tendermint block\-chain is a list called $\chain$ of headers that,
  for all $i < len(\chain) - 1$,
  satisfies: 
  \begin{enumerate}
\item 
 $\chain[i].Height + 1 = \chain[i+1].Height$
(We do not write $\chain[i].Height = i$, to allow that a chain
can be started at some arbitrary height, e.g., when there is social
consensus to restart a chain from a given height/block.)

\item
$\chain[i].\btime < \chain[i+1].\btime$

\item
$\chain[i+1].\vals = \chain[i].\nvals$
\label{item:valnextval}
\end{enumerate}
\end{definition}

Definition~\ref{def:basicinv}(\ref{item:valnextval}) captures the
     changing validator sets discussed in Section~\ref{sec:glance}.
In addition to these basic invariants there are invariants that are
     based on hashes and digital signatures.
We start to introduce their semantics by defining some preliminary
     functions: 

\begin{definition}[Abstract auxiliary soundness functions]\label{def:func}
The system provides the following functions:

\begin{enumerate}
  \item \label{item:hash}
 $\fhash$: We assume that every hash uniquely identifies the data it hashes

    \item
 $\fexecute$: used for state machine replication. The function maps $Data$
  (transactions) and an application state to a new state. It is a function
      (deterministic transitions).

\item
$\PossibleCommit$: There is a function $\PossibleCommit$ that maps a
  block (header) to the domain of $\commit$ from
  Definition~\ref{def:header}. 
\label{item:posscommit}
  \item \label{item:proof}
   $\fproof(b,commit)$: a predicate: true iff
      \begin{enumerate} 
     \item $b$ is in the $\chain$, i.e., there is an $i$ such that
       $\chain[i] = b$
     \item $commit$ is in $\PossibleCommit(b)$.
      \end{enumerate}
      
\end{enumerate}
\end{definition}

In Tendermint consensus~\cite{bkm2018latest}, the validators sign a
     given block.
A set of signatures by a quorum of validators for a block is called a
     commit.
Some of the required semantics can be proven independently of these
     details.
To capture these, we introduce
     Definition~\ref{def:func}(\ref{item:posscommit}).

Because $\fproof$ in Definition~\ref{def:func}(\ref{item:proof})
refers to the $\chain$, it depends on the execution, which results in a different
     quantifier order.
For instance, we say "there exists a function $\fhash$ such that for
     all runs", while we say "for each run there exists a function
     $\fproof$".
The consequence is that in Definition~\ref{def:func}(\ref{item:hash}) $\fhash$ is a predetermined function
     (implemented), while $\fproof$ will have to be computed at runtime as a function of the $\chain$.
The challenge in a distributed system is to locally compute $\fproof$
     without necessarily having complete knowledge of $\chain$.
In the context of the light client, we even want to infer knowledge
     about $\chain$ from the outcomes of the local computation of
     $\fproof$.
We  use digital signatures for that, and
 introduce them below when we introduce the distributed
     aspects.

\begin{definition}[Soundness predicates]
Given two blocks $b$ and $b'$:
\begin{enumerate}
\item $\fmatchhash(b,b')$ iff $\fhash(b) = b'.\blockid$
\item $\fmatchproof(b,b')$ iff $\fproof(b, b'.\commit)$
\end{enumerate}
\end{definition}

\begin{definition}[Security invariants] For all $i < len(\chain)-1$:
  \begin{enumerate}
    \item
  $\fmatchhash(\chain[i], \chain[i+1])$
 
   \item
 $\fmatchproof(\chain[i], \chain[i+1])$

   \item
 $\chain[i+1].AppState = $ \\ $\fexecute(\chain[i].Data,\chain[i].AppState)$
\end{enumerate}
\end{definition}

The function $\fmatchhash$ formalizes the hash arrow in
     Figure~\ref{fig:sequ-verif}.
We will show in Section~\ref{sec:checks} how $\fproof$ and thus
     $\fmatchproof$ can be checked in a distributed system in the
     presence of Authenticated Byzantine faults, based on quorums in
     $\votes$.
 
\section{Tendermint Security Model}
\label{sec:TMFM}

In Section~\ref{sec:intro}, we discussed that in Tendermint
     blockchains the proof-of-stake mechanism entails a time-dependent
     security model.
We capture this by the following formalization, which states that once
     a new validator set ($\nvals$) is chosen, we trust that it
     contains a correct quorum for some limited time, namely the
     trusting period.
We start with preliminaries.

\begin{definition}[Validator Data Structures]
  Given a full node, a \emph{validator pair} is a pair
  $(peerID, vp)$ , where  $peerID$ is the PeerID (public key) of a
full node, and the voting power $vp$ is an integer (representing the full node's
voting power in a certain consensus instance).
A \emph{validator set} is a set of validator pairs. For a validator set
$V$, we write $\totalvp(V)$ for the sum of the
voting powers of its validator pairs.
\end{definition}

\begin{definition}[Domain of Distributed Commit]\label{def:domain}
A commit is a set of \texttt{precommit} messages sent and signed by 
     validator nodes during the execution of Tendermint
     consensus\cite{bkm2018latest}.
Each message contains the following fields 
\begin{enumerate}
\item \texttt{Type}: precommit 
  \item
    \texttt{Height}: positive integer
    \item \texttt{Round} a positive integer 
\item \texttt{BlockID} a hash value of a block
\end{enumerate}
\end{definition}

We assume the authenticated Byzantine fault model~\cite{DLS88} in
     which no node (faulty or correct) may break digital signatures,
     but otherwise, no additional assumption is made about the
     internal behavior of faulty nodes.
That is, faulty nodes are only limited in that they cannot forge
messages. This implies for Definition~\ref{def:domain}, e.g.,
that a faulty node $p_f$ may sign a \texttt{precommit}
     message for a hash of a block that is not on the blockchain, but
     it may not generate a precommit message that appears to be signed
     by a correct node $p_c$ (unless $p_c$ actually signed that message
     before and $p_f$ received~it). 
     
A Tendermint blockchain has the \emph{trusting period}
  as a
     configuration parameter~$\trustp$.
We define a predicate $\correct(n, t)$, where $n$ is a node and $t$ is
     a time point.
The predicate $\correct(n, t)$ evaluates to true if and only if the node $n$
     follows all the protocols (at least) until time $t$. (It is false
     if a node $n$ deviates from the protocol once by time $t$.)

\begin{definition}[Security Model]
  \label{def:secmodel}
If a block $h$ is in the chain, then there exists a subset
$C$ of $h.\nvals$, such that:
$$
\totalvp(C) > \frac{2}{3}
\totalvp(h.\nvals)$$
and for
every validator pair $(n,p) \in C$, it holds that
$$\correct(n, h.\btime + \trustp).$$
 \end{definition}

The definition of correct refers to realtime, while it is used here
     with $\btime$ as stored in a block and the configuration
     parameter \emph{trusting Period} $\trustp$, which are "hardware
     times".
To not clutter the presentation, we do not make a distinction here
     between real-time and hardware time, and we assume that the
     hardware clock is sufficiently synchronized to real time.
Also, the trusting period $\trustp$ is typically in the order of
     weeks, so that inaccuracies in time synchronization can be dealt
     with security margins.

\begin{definition}[Distributed Commit]
 \label{def:distrcommit}
For a block $b$, each element $PC$ of $\PossibleCommit(b)$
satisfies that
\begin{enumerate}
  \item $PC$ contains only votes by validators from $b.\vals$
  \item $\totalvp(PC) > \frac{2}{3}
    \totalvp(b.\vals)$
    \item and  there is an $r$ such that
each vote $v$ in $PC$ satisfies:
\begin{enumerate}
  \item $v.\vtype = \mathsf{precommit}$
  \item $v.\bheight = b.\bheight$
  \item $v.\vround = r$  
  \item $v.\msgblockid = hash(b)$
\end{enumerate}
\end{enumerate}
\end{definition}

In a distributed commit necessarily all $\msgblockid$s
     are equal, which can be  checked locally.

We have now defined all the guarantees provided by a Tendermint
blockchain that are necessary to formalize what it means to observe
the state of the blockchain from the outside.
 
\section{The Light Client Verifier Problem}
\label{sec:problem}

In the most abstract viewpoint, the light client just implements a
     read of a header (block) of a given height from the blockchain.
This header $h$ needs to be generated by the Tendermint consensus.
In particular, a header that was not generated by the blockchain
     should never be stored.
Due to the evolving validator sets, without constantly following the
     progress of the blockchain, in the presence of Byzantine faulty
     nodes, one cannot know a priori who are the relevant validators
     for a block that are allowed to sign a block: For instance, a set
     of Byzantine nodes, which never participated in Tendermint
     consensus may generate and sign a block that structurally is
     according to the definitions.
Thus the \emph{Verifier} has to locally check whether the nodes who
     sign a block can be trusted, more precisely, whether sufficiently
     many correct nodes have signed.

We will start with the sequential problem statement that considers the
abstract case where the blockchain is just a data structure and there
are no faults, and we will then introduce the distributed model we
consider and the distributed problem statement.

\begin{definition}[Sequential Problem Statement]\label{def:sps}
  The Verifier satisfies the following properties
  \begin{description}
  \item[Safety.] The \emph{Verifier} never stores a
    header which is not in the
      blockchain.

\item[Liveness.] The \emph{Verifier} receives as input a height $\targeth$
     (not greater than the current height of the blockchain), and
     eventually stores the header of height $\targeth$ of the
     blockchain.
\end{description}
\end{definition}

\paragraph{Distributed Problem Statement}

To address the sequential problem statement, we consider the following
     setup:  The verifier communicates with a full node called
     \emph{primary}.
No assumption is made about the full node (it may be correct or
     faulty).
Communication between the light client and a correct full node is
     reliable and bounded in time.
Reliable communication means that messages are not lost, not
     duplicated, and eventually delivered.
There is a (known) end-to-end delay $\Delta$, such that if a message
     is sent at time $t$ then it is received and processes by time $t
     + \Delta$.
This implies that we need a timeout of at least $2 \Delta$ for a
     query/response communication (e.g., a remote procedure call) to
     ensure that the response of a correct peer arrives before the
     timeout expires.

As we do not assume that the {primary} is correct, no protocol can
     guarantee the combination of the sequential properties.
Thus, in the (unreliable) distributed setting, we consider two kinds
     of termination, \emph{successful} and \emph{failure}, and we will
     specify below under what (favorable) conditions the verifier can terminate successfully, and satisfy the
     requirements of the sequential problem statement.

\paragraph{Variables used by light client verification}

To formalize the problem, we need to define the state space of the
     protocol.
We do so by defining problem variables: the local data structure
     $\lightStore$  contains lightblocks that contain a header.
For each lightblock, we record its verification status, that is,
     whether it is verified.
The local variable $\primary$  contains the PeerID of a full node.
The container $\lightStore$ is initialized with a header
     $\trustedHeader$ that was correctly generated by Tendermint
     consensus.
We use the convention that the status of $\trustedHeader$ is verified.

\begin{definition}[Distributed Problem Statement]\label{def:dps}
The light client satisfies the following properties
  \begin{description}
    \item[safety.]
At all times, every verified header in $\lightStore$
was generated by an instance of Tendermint consensus.

\item[liveness.]
From time to time, a new instance of the verifier is called
with a height $\targeth$. Each instance must eventually
terminate.
 If the {primary} is correct, and $\lightStore$
always contains a verified header whose age is less than the trusting
period,
then the verifier adds a verified header $hd$ with
height $\targeth$ to $\lightStore$ and it
\textbf{terminates successfully}.
  \end{description}
  \end{definition}

These definitions imply that if the primary is faulty, a header may or
     may not be added to $\lightStore$.
The definition allows that verified headers are added to $\lightStore$
     whose height was not passed to the verifier (e.g., intermediate
     headers used in bisection; see Section~\ref{sec:proto}).
Note that for liveness, just initially  having a $\trustedHeader$
     within the trusting period is not sufficient.
For instance, if the trusting period expires before the first message
     round trip with the primary can be completed, the Tendermint
     security model does not provide any guarantees about correct and
     faulty nodes anymore.
After giving the specification of the protocol in
     Section~\ref{sec:proto}, we will discuss some liveness scenarios
     in Section~\ref{sec:liveness}.

\paragraph{Relation of the distributed to the sequential problem}

The specification in Definition~\ref{def:dps} provides a partial
     solution to the sequential specification in
     Definition~\ref{def:sps}.
The solution with respect to safety is complete, even if the primary
     is faulty.
However, we can only guarantee liveness when the primary is
     correct and the verifier has a sufficiently recent trusted
     header.
For these runs distributed liveness implies the sequential liveness.
Ensuring complete liveness (or perhaps just almost sure termination)
     would require us to make additional assumptions about the total
     (expected) number of faulty nodes in the
     network.
The security model imposes such assumptions on validator nodes only,
     which represent only a fraction of all the nodes in the
     Tendermint system (cf.\ Section~\ref{sec:intro}).
Adding incentives and punishment rules to nodes that communicate with
     light clients is subject of current discussion of the community,
     so that we cannot give reasonable additional assumptions in this
     paper.
However, in practice, if a run of the verifier fails, the light client
     may pick a new primary and retry until it reaches a correct
     primary which then ensures liveness. In this regard, it is assumed 
     that light clients have access to at least one correct full node.
 
\section{Light Client Verification}
\label{sec:checks}

The standard way of following the evolution of the blockchain is to
     download block after block and perform sequential verification as
     shown in Figure~\ref{fig:sequ-verif}.
Here we  discuss a verification method that does not
     force a client to download the headers for all blocks of height
     up to $\targeth$ (Definition~\ref{def:dps}).
The outline of the approach is given in Figure~\ref{fig:skip-verif}.
The method consists in asserting that the commit for the header of
     height~$\targeth$ contains the signature of at least one correct
     node.
This can be checked by exploiting the security model of
     Definition~\ref{def:secmodel}.
We have to consider the intersection of the set of validators in the
     commit and of correct nodes in a set $\nvals$ in a previously
     downloaded and trusted block.
By Definition~\ref{def:secmodel}, more than $2/3$ of the voting power
     in $\nvals$ is correct (for some time).
Now,  if the set of validators in a commit have more than $1/3$ of the
     voting power in $\nvals$, these two sets intersect, which implies
     that at least one validator is both (i) correct and (ii) signed
     the commit.
The following proposition establishes the part (i) of this argument,
     and is a direct consequence of Definition~\ref{def:secmodel}.

\begin{proposition}\label{prop:correct}
Given a (trusted) block \emph{tb} of the blockchain, at a real-time
$t$,
a given set of full
nodes \emph{N} contains a correct node, if
\begin{enumerate}
\item $t - \trustp < tb.Time < t$, and
  \item  the voting
power in $tb.\nvals$ of nodes in $N$ is more than $1/3$ of 
$\totalvp(tb.\nvals)$
\end{enumerate}
\end{proposition}

We now need to make explicit a property of the commits that comes from
     the way commits are computed by Tendermint consensus.
Analysis of the consensus algorithm in~\cite{bkm2018latest}
     immediately shows that a correct validator node only sends
     \texttt{prevote} or \texttt{precommit} messages  if
     \texttt{LastBlockID} of the new (to-be-decided) block is equal to
     the hash of the last block~$\ell$.
This implies that at a time where due to
     Definition~\ref{def:secmodel}, more than two thirds of
     $\ell.\nvals$  are still correct, we can trust a commit that is
     consistent with  $\ell.\nvals$.
Due to this, and by the fact that in the authenticated Byzantine model
     signatures cannot be forged we obtain the following proposition.

\begin{proposition}\label{prop:onthechain}
  Let $b$ be a block, and $c$ a commit. If at real-time~$t$
  \begin{enumerate}
    \item $c$ contains at least one
validator pair $(v,p)$ such that $v$ is correct\dash---that is,
$\mathit{correctUntil}(v,t)$\dash---, and
\item $c$ is contained in $\PossibleCommit(b)$
\end{enumerate}
then the block $b$ is on the blockchain.
\end{proposition}

The following central result is a direct consequence of
Propositions~\ref{prop:correct} and~\ref{prop:onthechain}

\begin{corollary}\label{cor:main}
Given a trusted block
$tb$ and a block $b$ with a commit $c$, at real-time $t$,
if
\begin{enumerate}
\item $t - \trustp < tb.Time < t$, and
  \item  the voting
power in $tb.\nvals$ of nodes in $c$ is more than $1/3$ of 
$\totalvp(tb.\nvals)$,
and
\item $c$ is contained in $\PossibleCommit(b)$,
\end{enumerate}
then the block $b$ is
on the blockchain.
\end{corollary}

As a result we need not resort to sequential verification, but can use
     the current time, and a block for which we have previously
     established trust to extend the trust to a new block.
However, if the preconditions Corollary~\ref{cor:main}(1--3) are not
     satisfied, this does not imply that $b$ is forged.
It might be that between $tb$ and $b$ the validator set has changed
     too much to ensure a sufficiently large intersection.
The protocol in the following section will then download an
     intermediate header whose height lies between $tb.Height$ and
     $b.Height$ and tries to get trust in the intermediate header and
     use this to eventually verify $b$.
 
\section{Protocol Description}
\label{sec:proto}

\begin{figure}[t]
    
\begin{lstlisting}[language=Go]
func VerifyToTarget(primary PeerID, lightStore LightStore, targetHeight Height) (LightStore, Result) {

nextHeight := targetHeight

for lightStore.LatestVerified().header.height < targetHeight {

    // Get Light Block
    current, found := lightStore.Get(nextHeight)
    if !found {
        current = FetchLightBlock(primary, nextHeight)
        lightStore.Update(current, StateUnverified)
    }

    // Verify
    verdict = ValidAndVerified(lightStore.LatestVerified(), current)
    
    // Decide where to continue
    if verdict == OK {
        lightStore.Update(current, StateVerified) <@\label{line:verf}@>
    }
    else if verdict == CANNOT_VERIFY {
      // do nothing. the light block current passed validation, 
      // but the validator set is too different to verify it.
      // We keep the state of current at StateUnverified. For a
      // later iteration, Schedule might decide to try
      // verification of that light block again.
    }    
    else { 
        // verdict is some error code
        lightStore.Update(current, StateFailed)
        return (lightStore, ResultFailure)
    } 
    nextHeight = Schedule(lightStore, nextHeight, targetHeight)
}
return (lightStore, ResultSuccess)
}
\end{lstlisting}
   \caption{Light Client Verification Main Function}
  \label{fig:main}
\end{figure}

The  basic data structure of our verification protocol is the
     $\lightStore$ which is a container for the so-called lightblocks,
     which correspond to the headers from Definition~\ref{def:header}.
In the implementation, lightblocks contain actual validator sets,
     while headers in the Tendermint implementation only contain
     hashes of these sets:     

\begin{definition}[Lightblock] The core data structure of the protocol
  is the \texttt{LightBlock}. It consists of the following fields:
\begin{itemize}
\item Header
\item Commit
  \item Validators
  \item NextValidators
  \item Provider
\end{itemize}
\end{definition}

The $\lightStore$ is a data structure that stores such lightblocks,
together with their state. The states are from the set $\{
StateUnverified,$ $StateVerified, StateFailed\}$.
The LightStore exposes the following functions to query stored
lightblocks.
\begin{description}
\item[\texttt{Get(height\ Height)\ (LightBlock,\ bool)}]
   returns a lightblock at a given height or
false in the second argument if the LightStore does not contain the
specified lightblock.

\item[\texttt{LatestVerified()\ LightBlock}]
  returns the highest verified lightblock.

\item[\texttt{Update(lightBlock\ LightBlock,\ v State)}]
    The\\ state of the lightblock is set to
\texttt{v}.
\end{description}

Our light client protocol is depicted in Figure~\ref{fig:main}.
It
gets as input
\begin{itemize}

\item
  $\lightStore$: a container that stores light  blocks that have been downloaded and
  that passed verification. Initially it contains a lightblock with
  $\trustedHeader$. As the function can be called multiple times,
  the lightStore may contain more lightblocks that have been
  downloaded and verified so far.
  
\item
  $primary$: the address (peerID) of the full node that the verification
  queries for blocks
\item
  $\targeth$: the height of the needed header. 
\end{itemize}

In this paper we consider the case where $\targeth$ is greater than
(or equal to)\linebreak $\lightStore.\LatestVerified().Header.Height$ as it is the
     most interesting.
In the other case there are two options, either there is a trusted
     lightblock (within the trusting period) with height less than
     $\targeth$, and we can use the same method as described here from
     that block, or we need to download all headers between $\targeth$
     and the height of a trusted lightblock in store and just check
     hashes in $\blockid$ in decreasing order of heights.

The function uses two auxiliary variables, namely $\nextHeight$, 
     which should be thought of as the
     ``height of the next header we need to download and verify'', and
     $\current$, the header that is currently under
     verification. $\nextHeight$ is initialized to $\targeth$.
     Then the protocol enters a loop that consists of
     the following three stages:

\begin{description}
\item[Get Lightblock] here a lightblock is assigned to $\current$. If the
  required lightblock had been downloaded before then it is taken from the
  $lightStore$, otherwise
  \texttt{FetchLightBlock} is called to download a lightblock 
  of a given height from the primary. This function is the only one
  that communicates with another node in the system.
\item[Verify]
  \texttt{ValidAndVerified} is local code that checks the lightblock. It
  encodes the checks of Corollary~\ref{cor:main} or, if sequential
  lightblocks should be verified falls back to standard sequential
  verification (cf.~Figure~\ref{fig:sequ-verif}). If it can verify
  $\current$ it returns \texttt{OK}. If the precondition of
  Corollary~\ref{cor:main} is violated (but otherwise  $\current$ is
  well-formed) it returns \texttt{CANNOT\_VERIFY}. Otherwise, that is,
  if the $\current$ has been proven to be corrupted, it returns an
  error code.
\item[Decide where to continue]
  \texttt{Schedule} decides which height to try to verify next. We keep
  this underspecified as different implementations (currently in Golang
  and Rust) may implement different optimizations here. 
\end{description}

For \texttt{Schedule},  we 
  provide the following
  necessary conditions on how the height may evolve: \texttt{Schedule}
  returns $H$ s.t.
  \begin{enumerate}
      \item[(S1)]\label{itm:S1} if
$lightStore.\LatestVerified().Header.Height = \nextHeight$ and\\
$lightStore.\LatestVerified().Header.Height < \targeth$ then\\
$\nextHeight < H \le  targetHeight$
\item[(S2)]\label{itm:S2} if
$lightStore.\LatestVerified().Header.Height < \nextHeight$ and\\
$lightStore.\LatestVerified().Header.Height < \targeth$ then\\
$lightStore.\LatestVerified().Header.Height < H <
  \nextHeight$
\item[(S3)]\label{itm:S3} if $lightStore.\LatestVerified().Header.Height = targetHeight$
then
$H = targetHeight$
  \end{enumerate}

  Case~(S1) captures the case where the lightblock at height
     $\nextHeight$ has been verified, and we can choose a height closer
     to  $\targeth$.
As \texttt{Schedule} gets the $lightStore$ as parameter, the choice
     of the next height can depend on the $lightStore$, e.g., we can
     pick a height for which we have already downloaded a lightblock.
Case (S3) is a special case when we have verified $\targeth$.
In Case (S2) the lightblock of $\nextHeight$ could not be verified, and we
     need to pick a smaller height.

\paragraph{Invariant}
The implementation enforces the invariant that it is always the case
that
\begin{equation}
  \lightStore.\LatestVerified().Header.Time > now -
  \trustp.
  \label{eq:inv}
  \end{equation}
If the invariant is violated, the light client does not have a lightblock
     it can trust and it terminates with failure.
A trusted lightblock must be obtained externally, its trust can only be
     based on social consensus.

\subsection{Correctness}

     \begin{proposition} The protocol satisfies safety.
\end{proposition}

     \begin{proof}
It is sufficient to remark, that a lightblock is marked as
verified in line~\ref{line:verf} if \texttt{ValidAndVerified} returned
\texttt{OK}, which is the case only if the preconditions of
Corollary~\ref{cor:main}  are satisfied (in which case safety is
ensured by the corollary) or if we fall back to sequential
verification in which case Definition~\ref{def:distrcommit} is checked.
     \end{proof}

     \begin{proposition} The protocol satisfies liveness.
\end{proposition}

     \begin{proof}
We proof by case distinction regarding the primary:

\begin{description}

\item[If the {primary} is correct] then

  \begin{itemize}
  
  \item
    \texttt{FetchLightBlock} will always return a lightblock consistent
    with the blockchain
  \item
    \texttt{ValidAndVerified} will verify a correct lightblock once a
    sufficiently recent lower lightblock can be verified.
  \item
    If Invariant~(\ref{eq:inv}) holds,
    eventually every lightblock will be verified and core verification
    \textbf{terminates successfully}.
  \item
    As by Definition~\ref{def:dps}, if the primary is correct,
    for liveness we are restricted to
   the case when  Invariant~(\ref{eq:inv}) holds, we concludes the proof.
  \end{itemize}

\item[  If the {primary} is faulty] then there are three cases:

  \begin{itemize}
  
  \item
    it either provides lightblocks in time that pass all the tests, and
    the function returns
    with the lightblock
  \item
    or it provides one lightblock that fails a test, and the function
    \textbf{terminates with failure}.
    \item or it is too slow in (or stops) providing lightblocks, such that
      eventually 
      Invariant~(\ref{eq:inv}) is discovered to be violated,
      and the protocol terminates
      with failure.
  \end{itemize}
\end{description}
This concludes the liveness argument.
\end{proof}

\subsection{Liveness Scenarios} \label{sec:liveness}
     
The simplicity in the above liveness proofs is due to
     Invariant~(\ref{eq:inv}).
The problem definition allows that a protocol does nothing: Once the
     invariant is violated we are allowed to terminate with a failure.
Successful termination depends on the age of
     $lightStore.\LatestVerified()$ (for instance, initially on the age
     of $\trustedHeader$) and the changes of the validator sets on the
     blockchain.
We will now give some examples.

Let $\startHeader$ be $\lightStore.\LatestVerified()$ when core
verification is called (e.g., $\trustedHeader$) and $\startTime$ be
the time the verifier is invoked.

In order to ensure liveness, $\lightStore$ always needs to contain a
verified (or initially trusted) lightblock whose time is within the trusting
period. To ensure this, the verifier needs to add new lightblocks to
$\lightStore$ and verify them, before all lightblocks in
$\lightStore$ expire.

\paragraph{Many changes in validator set}

Let's consider \texttt{Scheduler} implements bisection, that is, it halves
the distance. Assume the case where the validator set changes completely
in each block. Then the method in this specification needs to
sequentially verify all lightblocks. That is, for
$W = \log_2 (\targeth - \startHeader.\bheight)$,
$W$ lightblocks need to be downloaded and checked before the lightblock of
height $\startHeader.Height + 1$ is added to $\lightStore$.

\begin{itemize}

\item
  Let $Comp$ be the local computation time needed to check lightblocks
  and signatures for one lightblock.
\item
  Then we need in the worst case $Comp + 2 \Delta$ to download and
  check one lightblock.
\item
  Then the first time a verified lightblock could be added to
  $\lightStore$ is $\startTime + W  (Comp + 2 \Delta)$.
\item
   However, it can only be added if we still have a lightblock in
  $\lightStore$, which is not expired, that is only the case if

  \begin{itemize}
  
  \item
    $\startHeader.Time < \startTime + W  (Comp + 2 \Delta) -
    \trustp$,
  \item
    that is, if core verification is started at
    $\startTime < \startHeader.Time + \trustp - W 
    (Comp + 2 \Delta)$
  \end{itemize}

\end{itemize}

Starting from the above argument  one may then do an inductive
     argument from this point on, depending on the implementation of
     \texttt{Schedule}.
We may have to account for the lightblocks that are already downloaded,
     but they are checked against the new
     $\lightStore.\LatestVerified()$.

We observe that the worst case time it needs to verify the lightblock of
     height $\targeth$ depends mainly on how frequent the validator
     set on the blockchain changes.
That the verifier terminates successfully crucially depends on
     the check that the lightblocks in $\lightStore$ do not expire in the
     time needed to download more lightblocks, which depends on the
     creation time of the lightblocks in $\lightStore$.
That is, termination of the verifier is highly depending on the
     data stored in the blockchain.
The current light client  verifier protocol exploits that in practice
     changes in the validator set are rare.
For instance, consider the following scenario.

\paragraph{No change in validator set}

Assume that on the blockchain the validator set of the block at height
     $\targeth$ is equal to $\startHeader.\nvals$.
Then there is one round trip in \texttt{FetchLightBlock} to download
     the lightblock of height $\targeth$, and $Comp$ to check it.
As the validator sets are equal, \texttt{ValidAndVerified} returns
     \texttt{OK}, if $\startHeader.Time > now - \trustp$.
That is, if $\startTime < \startHeader.Header.Time + \trustp -
     2 \Delta - Comp$, then the verifier terminates successfully.

\section{Formalization in \tlap{}}
\label{sec:tla}

As part of our formalization efforts, we have specified the light client
    protocol in~\tlap{} and checked its properties with the symbolic model
    checker~\textsc{Apalache}~\cite{KKT19}. 
We found that~\tlap{} allows us to express the protocol at the level that is
    quite close to the mathematical description, which we provide in
    Sections~\ref{sec:blockchain}--\ref{sec:proto}. 
In addition to having a machine-readable protocol specification, we were able
    to produce non-trivial system executions with the model checker as well as
    verify the protocol properties for small parameter values.
The complete specification can be found in Appendix~\ref{sec:tla-complete}.
In this section, we highlight non-obvious decisions about our specification.

Our~\tlap{} specification consists of several building blocks:
    the reference chain, the primary model, and the protocol specification.
The reference chain is populated before the light client runs.
Depending on whether the primary peer is correct or faulty,
 the communication with the primary is modelled as a non-deterministic action
    that either copies blocks from the reference chain, or it produces
    corrupted blocks.

Our specification has five parameters:
\begin{lstlisting}[language=tla]
CONSTANTS
  AllNodes,           \* a set of potential validators
  IS_PRIMARY_CORRECT, \* is primary correct (a Boolean)
  TRUSTING_PERIOD,  \* trusting period in discrete time units
  TRUSTED_HEIGHT,   \* the starting height of the client
  TARGET_HEIGHT     \* the goal height of the client
\end{lstlisting}

By fixing the specification parameters, we can verify the protocol properties
    with the model checker and observe counterexamples to the false hypotheses.

\subsection{Specifying the reference chain}

The reference chain is simply a function from block heights to the lightblocks. 
Since the model checker supports only finite sets, we limit the domain of this
    function to the set $1..\textsf{TARGET\_HEIGHT}+1$. 
To this end, we first define the sets of block headers and lightblocks:

\begin{lstlisting}[language=tla]
BlockHeaders $\tladef$[ height: 1.. (TARGET_HEIGHT + 1),
                 time: Int,
                 VS: SUBSET AllNodes,
                 NextVS: SUBSET AllNodes,
                 lastCommit: SUBSET AllNodes ]
LightBlocks $\tladef$
  [header: BlockHeaders, Commits: SUBSET AllNodes]
\end{lstlisting}

In~\tlap{}, notation~$[a: A, \dots, z: Z]$ defines the set of records whose
    fields~$a,\dots,z$ are restricted to the sets~$A, \dots, Z$, respectively.
Moreover, $\textsc{SUBSET}~X$ defines the powerset of a set~$X$.

Several comments about the block headers are in order. 
First, we model timestamps as integers, leaving the time resolution up to
    the user's interpretation. 
Second, we do not explicitly model digital signatures and thus model the
    validator sets and commits as subsets of~\textsf{AllNodes}. 
Third, we omit hashes and instead limit the power of faulty peers in the peer
    model.
Although we could add hashes in the specification, we found that they do not
    improve the protocol understanding, as they are an implementation detail.
Fourth, we restrict  voting powers to~$\{0, 1\}$; otherwise, we would
    have to use multisets instead of sets. 
(One can model a validator with a voting power of~$k$ with~$k$ validators
    with the voting power of~1.)

We define the predicate~\textsf{InitToHeight} that restricts the
    function~\textsf{blockchain} as per Definitions~\ref{def:basicinv}
    and~\ref{def:secmodel}. 
This predicate also non-deterministically selects a set of faulty
    validators~$\mathsf{Faulty} \subseteq \mathsf{AllNodes}$ and a value of the
    global clock~$\mathsf{now}$, which must be above the timestamp of the
    last block: $\mathsf{now} \ge
    blockchain[1+\mathsf{TARGET\_HEIGHT}].time$.
Interestingly, we model a global clock as an integer variable, by following
    the Lamport's approach~\cite{lamport2005real-time}.

\subsection{Specifying the primary and light client}

The light client maintains the following state variables:

\begin{lstlisting}[language=tla]
VARIABLES state, nextHeight,
  fetchedLightBlocks, lightBlockStatus, latestVerified
\end{lstlisting}

These variables are similar to those in Figure~\ref{fig:main}. 
The variable~$\mathsf{state}$ encodes the
    progress of the light client and ranges over $\{``\mathit{working}",$
    $``\mathit{finishedSuccess}", ``\mathit{finishedFailure}"\}$.
The variable~$\mathsf{nextHeight}$ is as in
    Figure~\ref{fig:main}. 
The other three variables model the lightstore:
The variable~$\mathsf{fetchedLightBlocks}$ is a function from a subset of
    heights to~$\mathsf{LightBlocks}$, which maintains the lightblocks
    received from the primary; $\mathsf{lightBlockStatus}$ maps those heights
    to the states ``\textsf{StateVerified}'', ``\textsf{StateUnverified}'',
    and ``\textsf{StateFailed}''. 
Finally, \textsf{latestVerified} maintains a copy of the latest verified block.

We encode a system transition with the predicate~\textsf{Next} as follows:

\begin{lstlisting}[language=tla]
Next $\tladef$
 $\land$ state = "working" $\land$ (VerifyToTargetDone $\lor$ VerifyToTargetLoop)
 $\land$ $\exists$ t $\in$ Int: t $\ge$ now $\land$ now' = t
 $\land$ UNCHANGED $\langle$blockchain, Faulty$\rangle$
\end{lstlisting}

In~\textsf{Next}, the light client either performs one iteration of the loop in
    Figure~\ref{fig:main} (by performing action~\textsf{VerifyToTargetLoop}),
    or terminates the loop (by performing action~\textsf{VerifyToTargetDone}).
Simultaneously, the global clock~\textsf{now} advances by a non-negative value.
We omit the details of~\textsf{VerifyToTargetLoop}
    and~\textsf{VerifyToTargetDone}, as they closely follow the code
    in Figure~\ref{fig:main}.

\begin{table*}
  \caption{\boldmath{Model checking experiments with~\textsc{Apalache}. A configuration $n/k/B$
    represents $n$ validator nodes and $k$ blocks.
    The primary is correct when $B=C$, and the primary is faulty when $B=F$.
    }}
  \label{tab:mc}
      \begin{tabular}{l|lr|lr|lr|lr|lr}
        \hline
            \tbh{Property}
                   & \multicolumn{2}{c|}{\tbh{4/3/C}}
                   & \multicolumn{2}{c|}{\tbh{4/3/F}}
                   & \multicolumn{2}{c|}{\tbh{5/5/C}}
                   & \multicolumn{2}{c|}{\tbh{5/5/F}}
                   & \multicolumn{2}{c}{\tbh{7/5/F}}
            \\\hline
            & \tbh{result} & \tbh{time}
            & \tbh{result} & \tbh{time}
            & \tbh{result} & \tbh{time}
            & \tbh{result} & \tbh{time}
            & \tbh{result} & \tbh{time}
            \\\hline
            \tbh{PositiveBeforeTrustedHeaderExpires}
            & \bmark{1} & 9s
            & \bmark{1} & 9s
            & \bmark{1} & 6s
            & \bmark{1} & 6s
            & \bmark{1} & 8s
            \\
            \tbh{Correctness}
                & \bcheck{4} & 9s
                & \bcheck{4} & 9s
                & \bcheck{11} & 3m46s
                & \bcheck{11} & 5m28s
                & \bcheck{11} & 13m20s
            \\
            \tbh{Precision}
                & \bcheck{4} & 9s
                & \bcheck{4} & 8s
                & \bcheck{11} & 2m38s
                & \bcheck{11} & 2m51s
                & \bcheck{11} & 4m35s
            \\
            \tbh{SuccessOnCorrPrimaryAndChainOfTrust}
                & \bcheck{4} & 9s
                & \bcheck{4} & 8s
                & \bcheck{11} & 2m48s
                & \bcheck{11} & 2m2s
                & \bcheck{11} & 3m28s
            \\
            \tbh{NoFailedBlocksOnSuccess}
                & \bcheck{4} & 9s
                & \bcheck{4} & 10s
                & \bcheck{11} & 2m14s
                & \bcheck{11} & 2m4s
                & \bcheck{11} & 3m25s
            \\
            \tbh{StoredHeadersAreVerifiedOrNotTrusted}
                & \bcheck{4} & 10s
                & \bcheck{4} & 10s
                & \bmark{4} & 17s
                & \bmark{4} & 16s
                & \bmark{4} & 23s
            \\
            \tbh{CorrectPrimaryAndTimeliness}
                & \bcheck{4} & 9s
                & \bcheck{4} & 8s
                & \bcheck{11} & 2m46s
                & \bcheck{11} & 2m
                & \bcheck{11} & 3m10s
            \\
            \tbh{Complexity}
                & \bcheck{4} & 9s
                & \bcheck{4} & 8s
                & \bcheck{11} & 2m8s
                & \bcheck{11} & 2m4s
                & \bcheck{11} & 3m52s
            \\\hline
     \end{tabular}

 \end{table*}

The behavior of a primary node is captured by the operator
    \textsf{FetchLightBlockInto}, which is shown below:

\begin{lstlisting}[language=tla]
CopyLightBlockFromChain(block, height) $\tladef$
  LET refBlock $\tladef$ blockchain[height] IN
  LET lastCommit $\tladef$ blockchain[height + 1].lastCommit IN
  block = [header $\mapsto$ refBlock, Commits $\mapsto$ lastCommit]      

IsLightBlockAllowedByDigitalSignatures(height, block) $\tladef$ 
  \* either the block is produced by consensus
  \*  (enforced by hashes), while commits are not restricted
  $\lor$ block.header = blockchain[height]
  \* or the block is signed only by the faulty validators
  $\lor$ block.Commits $\subseteq$ Faulty $\land$ block.header.height = height

FetchLightBlockInto(block, height) $\tladef$
  IF IS_PRIMARY_CORRECT
  THEN CopyLightBlockFromChain(block, height)
  ELSE IsLightBlockAllowedByDigitalSignatures(height, block)
\end{lstlisting}

In this code, a correct primary simply copies the block header and the
    respective commit from the reference chain. 
A faulty peer has more freedom, which is restricted
    with the predicate~\textsf{IsLightBlockAllowedByDigitalSignatures}. 
Like a correct primary, it can also produce a sound lightblock. 
Additionally, it may produce a sound block header, but an incorrect set of
commits. 
Alternatively, if the block header is different from the block on the reference
    chain, then it may be signed only by the faulty validators.

\subsection{Model Checking Experiments}

Our main goal in the verification efforts is to prove safety and liveness of
    the protocol as per Definition~\ref{def:dps}.
We have formalized the safety property of Definition~\ref{def:dps} as a state invariant
    called~\textsf{Correctness}.
Assuming that the protocol terminates, the liveness property of Definition~\ref{def:dps}
    can also be written as a state invariant, which describes the state upon termination
    (successful or not).
We call this state invariant \textsf{SuccessOnCorrPrimaryAndChainOfTrust}.

We perform bounded model checking with~\textsc{Apalache}, which explores
    executions up to a given length. 
Although this activity can produce counterexamples, it does not guarantee
    absence of bugs. 
Interestingly, the light client always terminates in a fixed number of steps
    that depends on the difference between the target height and the trusted
    height. 
If we call this difference~$\delta$, then the protocol should terminate in no
    more than $T(\delta)=\frac{\delta \cdot (\delta - 1)}{2}$ iterations. 
This is the worst-case bound for conditions (S1)--(S3), and a concrete
    implementation may schedule block queries more optimally, e.g., a
    worst-case linear bound or an expected sublinear bound.

To check the complexity  bound, we have written a state invariant called
    \textsf{Complexity}, which tests that the protocol does not go over the
    worst-case bound. 
When we fix the specification parameters, \textsc{Apalache} finds a deadlock
    for the computations longer than $T(\delta)$. 
Together with~\textsf{Complexity}, this gives us a termination argument (for
    fixed parameters).
Hence, it suffices to check the above invariants.     

To improve our understanding of the protocol, we have also specified a
     few additional properties.
For instance, one could (wrongly) expect that the initially trusted
     block should still be within the trusting period, when the light client
     terminates.
We formulate this property as a state invariant below:  

\begin{lstlisting}[language=tla]
PositiveBeforeTrustedHeaderExpires $\tladef$
 LET trustedTime $\tladef$ blockchain[TRUSTED_HEIGHT].time IN
 state = "finishedSuccess"
  $\Rightarrow$ trustedTime $\ge$ now + TRUSTING_PERIOD
\end{lstlisting}

This invariant is violated, as the light client trusts the block at
    \textsf{TARGET\_HEIGHT} upon successful termination, whereas the block at
    \textsf{TRUSTED\_HEIGHT} may be not trusted anymore.
The model checker produces a counterexample in one step.

Another false hypothesis is formulated in the invariant candidate
    \textsf{StoredHeadersAreVerifiedOrNotTrusted}. 
It is a weaker version of \textsf{SuccessOnCorrPrimaryAndChainOfTrust}, as it
    is inspecting all blocks, not only the verified ones. 
Although this property holds for $\delta=3$, it fails for $\delta=5$. 
The model checker is showing us a counterexample, where the lightblock~5 is
    verified on the basis of block~3, while block~4 is kept unverified.

\begin{lstlisting}[language=tla]
StoredHeadersAreVerifiedOrNotTrusted $\tladef$
 state = "finishedSuccess"
  $\Rightarrow$ $\forall$ lh, rh $\in$ DOMAIN fetchedLightBlocks:
   $\lor$ lh $\ge$ rh
   $\lor$ $\exists$ mh $\in$ DOMAIN fetchedLightBlocks: lh < mh $\land$ mh < rh
   $\lor$ LET l = fetchedLightBlocks[lh]
         r = fetchedLightBlocks[rh] IN
     $\lor$ "OK" = ValidAndVerified(l, r)
     $\lor$ now - l.header.time > TRUSTING_PERIOD
\end{lstlisting}

Table~\ref{tab:mc} summarizes the results of our experiments
    with~\textsc{Apalache}. 
The experiments were run in an AWS instance equipped with 32GB RAM and a 4-core
    Intel\textsuperscript{\textregistered}
    Xeon\textsuperscript{\textregistered} CPU E5-2686 v4 @ 2.30GHz CPU. 
We write ``\bmark{k}'' when a bug is reported at depth~$k$,
    and ``\bcheck{k}'' when no bug is reported up to depth~$k$.
We ran the experiments for small parameter values such as 4 to 7 nodes and 3 to
    5 blocks.
For these parameters, the tool responds in a matter of minutes, which is
    fast enough for us to see interesting counterexamples. 
For larger parameter values, the model checking gets significantly
    slower. 
We believe that this is caused by powersets and cardinality tests.

We also tried to run the model checker~\textsc{TLC}. 
However, we ran into two problems. 
First, \textsc{TLC} enumerates states, so it requires timestamps and the global clock 
    to range over a finite domain.
Second, even when we introduced a logical abstraction of time,
    \textsc{TLC} could not inspect all initial states, as it had to
    enumerate all combinations of multiple powersets.
Although we believe that it would be possible to use this model checker by
    introducing a more abstract version of the blockchain, we found
    that~\textsc{Apalache} was sufficient for our purposes.

In conclusion, model checking has improved our understanding of
    the protocol. 
It also has confirmed our intuition by showing us counterexamples. 
In the future, we plan to construct an inductive invariant, to obtain
    a complete argument over block heights.

\section{Implementation}
\label{sec:impl}

We implemented the
light client verification protocol in the Rust programming language.

\paragraph{Architecture}

The light client is architected for composability:
It was expected that a light client running the verification
     algorithm would share the process space with other components
     running in separate threads.
These separate components require a synchronous interface to fetch
     headers  the light client verified.

To maintain simplicity of the core verification logic, the
     light client is implemented as a finite state machine operating
     on events fetched from an unbounded queue.
Events processed by the light client can then be understood as atomic
     transformations, performed on state owned and encapsulated within
     the light client.
Events sent to the light client can include a callback to facilitate
     synchronizing interactions between components.
The queues in this case have the benefit of serializing all access to
     the light client core logic, eliminating the need for mutexes while
     guaranteeing memory safety.

Interactions with the light client are performed via a facade which
     acts as a thin interface exposing synchronous methods which
     serialize interactions with the light client runtime via the
     queue.
The method set for this interface is abstracted in such a way as to
     allow mock replacements to be used during testing.
This abstraction allows testing complex interactions between
     components.

\paragraph{Rust implementation}

Each aspect of the protocol is specified at the code level as an
     interface (called \texttt{trait} in Rust), hereafter called a
     \emph{component}.
Each component may (but need not) depend on  other components.
This allows us to unit-test each component independently by mocking
     out the others components it depends on.
Moreover, this approach also enables us to implement deterministic and
     reproducible tests by mocking out components which perform
     intrinsically non-deterministic computations, such as performing
     network requests or fetching the current system time.

Figure~\ref{fig:rust-verifier-trait} shows the definition of the
Verifier trait, which consists of a single method taking in a untrusted block,
a trusted one, and a set of options including the $\trustp$ and the current time.
The definition of a concrete verifier which depends on other components
is provided in Figure~\ref{fig:rust-verifier-def}.

The implementation spans around 2500 lines of code (not counting
     comments and whitespace), and is openly available~\cite{impl}.

\begin{figure}[t]
\begin{minipage}[t]{0.4\textwidth}
\begin{lstlisting}[language=Rust]
pub trait Verifier {
    fn verify(
        &self,
        untrusted: &LightBlock,
        trusted: &LightBlock,
        options: &Options
    ) -> Verdict;
}
\end{lstlisting}
\end{minipage}
\caption{Definition of the verifier trait}
\label{fig:rust-verifier-trait}
\end{figure}

\begin{figure}[t]
\begin{minipage}[t]{0.4\textwidth}
\begin{lstlisting}[language=Rust]
pub struct ProdVerifier {
    predicates: Box<dyn VerificationPredicates>,
    voting_power_calculator: Box<dyn VotingPowerCalculator>,
    commit_validator: Box<dyn CommitValidator>,
    header_hasher: Box<dyn HeaderHasher>,
}
\end{lstlisting}
\end{minipage}
\caption{Definition of a concrete verifier}
\label{fig:rust-verifier-def}
\end{figure}

\section{Related Work}

Bitcoin introduced the notion of a light client protocol in the form 
    of \emph{simplified payment verification (SPV)}~\cite{nakamoto2019bitcoin}.
    In SPV, a client downloads complete chains of block headers in order to discover
    the longest chain, or more accurately, the chain with the greatest amount of computational work,
    which is deemed to be the canonical one. From there, it can verify proofs of transaction 
    inclusion in any of the blocks. Notably, the protocol is linear in the
    number of blocks, which may be prohibitive for long chains.
    Sublinear variants of SPV have been proposed, including so-called Proofs-of-Proofs-of-Work
    ~\cite{kiayias2016proofs, kiayias2017non} and
    Flyclient~\cite{bunz2019flyclient}, which utilize probabilistic sampling to reduce
    the number of headers a light client must download. These solutions apply strictly to consensus protocols
    where agreement is determined by a heaviest-chain scoring metric, and are
    thus not relevant to BFT protocols like Tendermint,
    where chains are extended one block at a time by a quorum of validators.

Tendermint was the first system to lift traditional BFT consensus
     protocols~\cite{CastroL02,DLS88,Lamport11a,BessaniSA14} into the
     blockchain domain~\cite{kwon2014tendermint}.
In traditional BFT consensus, clients submit requests directly to
     validator nodes (known as \emph{replicas}), and wait to receive a
     quorum of identical responses ~\cite{CastroL02, BessaniSA14}.
That is, these systems expect clients to know the network addresses
     of validators and to maintain direct connections with them.
And while some do not even support validator set changes (i.e.,
\emph{reconfiguration})~\cite{CastroL02}, those that do expect clients
     to learn about the latest validator set from some unspecified
     directory service~\cite{BessaniSA14}.
In a public, open-membership, adversarial setting with arbitrary
     validator set changes and no trusted directory service, such an
     approach to servicing clients is wholly insufficient.

Most comparable blockchain systems avoid this problem by restricting
     validator set changes to happen in ``epochs'', so that the set of
     validators is static for a period of time (an epoch) and can only
     change at the epoch boundary.
In this setting, a lightclient could always skip from the first to
     last height in the epoch, and only needs to verify validator set
     changes at the epoch boundaries.
Tendermint, however, does not restrict the changes of the validator
     sets; they can happen at every block.
     
Tendermint emerged in the context of Proof-of-Stake blockchains, where
     economic stake within the system, rather than resource
     consumption outside the system, is used to incentivize correct
     behaviour.
Proof-of-Stake systems have long been known to suffer from the
     so-called \emph{nothing-at-stake} attack, whereby past validators
     who have since exited their stake can forge arbitrary alternative
     histories~\cite{subjectivity}.
Such attacks are solved by \emph{subjective
     initialization}, whereby a client subjectively decides which
     validators to initially trust, and by an \emph{unbonding period},
     during which validators can be punished for misbehaviour, and
     beyond which they can no longer be trusted by clients.
While much has been written about this informally
     online~\cite{subjectivity, pos_clients, casper_tendermint,
     eth_pos_faq}, we are not aware of a formal treatment of  the
     Proof-of-Stake light client problem.

From the viewpoint of more classic research literature, the light
     client  problem is a modern variant of performing a
     read operation from a replicated database~\cite{BernsteinHG87},
     or reading
     a shared state~\cite{AttiyaBD95}, when some of the peers are
     faulty, or learning an accepted value in
     Paxos~\cite{generalizedpaxos}.
As this is an important problem, there is a vast literature on this
     subject, also with respect to diverse consistency criteria,
     e.g.,~\cite{GotsmanB17}.
As Tendermint blockchains provide ``immediate
     finality''\dash---a.k.a., irrevocability in more classic
     consensus definitions~\cite{Charron-BostS09}\dash---, for the
     light client we are interested in strong consistency.
     
The first contribution of this paper is a formalization of the
     Tendermint Security model.
It shares the aspect of Authenticated Byzantine Faults with in the
     classic work in~\cite{DLS88,CastroL02}.
However, the staking mechanism requires us to formalize a notion close
     to Byzantine faults with recovery which is less
     studied~\cite{CastroL02,BarakHHN00,ADFHW07:OPODIS}. A similar concept has
     also be considered for communication faults~\cite{BielyWCGHS07}.

Other approaches to achieve a sublinear traversal of blockchains include the use of
    alternative authenticated data structures or more advanced cryptography.
    For instance, the Merkle Mountain Ranges \cite{CrosbyW09,mmrs} used in 
    Flyclient \cite{bunz2019flyclient} can be used in a BFT-based blockchain for
    logarithmically verifying that a past block is an ancestor of a more recent trusted block.
    However, as noted, they cannot be used for verifying that a future block is a child
    of a past trusted block without a mechanism like that described in this paper.
    Skipchains \cite{nikitin2017chainiac} do allow clients
    to skip from past to future blocks, though they require the retroactive addition 
    of (aggregated or collective) signatures to past blocks. Since past blocks cannot 
    be directly modified, such protocols should be considered as services layered on top of 
    the underlying blockchain protocol. Finally, recent advances in cryptography
    ~\cite{groth2016size, ben2017scalable} enable blockchain designs where clients 
	verify succint proofs attesting to some set
	of state transitions. This can take the form of proofs that the validator set changed in
	a particular way~\cite{gabizon2020plumo}, or, in a more extreme case, 
	proofs attesting to the correct 
	execution of the entire blockchain protocol, which eliminate the need to traverse the 
	chain at all~\cite{meckler2018coda}. While exciting,
	such protocols tend to require more exotic assumptions beyond the standard 
	authenticated Byzantine fault model and are thus less proven in real-world systems.

While we have focused here on a sublinear light client protocol for
     verifying a Tendermint blockchain under the Security Model we
     outlined, we have not addressed  what guarantees remain in the
     event the security model fails, i.e., when 1/3 or more of the
     voting power is faulty.
While such a scenario may cause our light client to accept a faulty
     chain (i.e., a fork), such faults may be detected, so long as the
     client is connected to at least one honest full node\dash---a standard
     assumption among light client protocols.
Furthermore, in future work, we intend to show that validators are
     accountable, that is, detection of forks will result in the
     faulty validators being identified, and thus punished
     accordingly.
While such protocols may detect forks in the blockchain, other
     protocols have focused on detecting invalid state transitions,
     where validators commit to a state that cannot be derived from
     applying transactions in the blockchain to the previous
     state~\cite{albassam2018fraud}.
Such protocols are complementary to ours; they use so-called \emph{fraud
     proofs} and \emph{data-availability proofs} to allow  light clients to
     detect invalid state transitions, even when a majority of
     validators are faulty.

\newcommand{\manycites}{\cite{BerkovitsLLPS19,DHVWZ14,KLVW17:POPL,GGB19,KraglQH18,BouajjaniEJQ18,DDMW19:CAV}}

We have used \tlap{}~\cite{lamport2002specifying} for  specification as
     it became a Lingua Franca for formal specification of complex
     distributed systems~\cite{NewcombeRZMBD15}.
The APALACHE model checker~\cite{KKT19} proved very
     effective to model check the protocol under different fault
     scenarios.
Recently there has been made significant progress in automated
     verification of fault-tolerant distributed algorithms~\manycites.
While this work typically focuses on consensus and Paxos-like
     algorithms, our work considers how to observe the state of the
     result of consensus from the outside.
In systems that solve consensus, the notions of quorums or thresholds
     are crucial and at the core of verification
     approaches~\cite{DHVWZ14,KLVW17:POPL,BerkovitsLLPS19}.
     In our system, these quorums appear in limited form, namely as data
     in the validator
     sets and commits.
For now, our current model checking results just consider small
     systems (up to seven validators), but we are confident to be able
     to adapt the recent results in automated verification to our
     domain in order to be able to scale to realistic sizes, or even
     to the parameterized case~\cite{AK86,EN95,2015Bloem}, that is, for all
     numbers of validators.

Recently, Ognjanovic~\cite{Ogn2020} used our \tlap{} model as a base
     for implementing the light client verification protocol in Scala
     and verifying it with Stainless
     (\url{https://stainless.epfl.ch}).

\section{Conclusions}
\label{sec:conc}

\paragraph{Security.}
We have presented the first formalization of the Tendermint security
     model, which allows us to understand it as an Authenticated
     Byzantine model with dynamic Byzantine faults.
We presented a light client verification protocol based on that model,
     and proved that it is always safe, and that it satisfies liveness
     if it communicates with a correct full node.
It is clear that in principle faulty full nodes would benefit from
     lying to the light client, by trying to make the light client
     accept a block that deviates (e.g., contains additional
     transactions) from the one generated by Tendermint consensus.
However, our safety properties guarantees that this cannot happen if
     the security model holds.

However, the question remains whether for liveness, full nodes would
     benefit from cooperating, i.e., from responding timely.
This is indeed the case if we consider the broader context where the
     classic less than $1/3$ model may be violated.
In this case,  the light client may help the correct full nodes to
     understand whether their header is a good one, or in other words,
     to detect forks on the chain.
In parallel to the verification logic described in this paper, we also
     design a fork detector that probes multiple full nodes.
In combination with the detector, the correct full nodes indeed have
     the incentive  to respond, and we can base our liveness arguments
     on the assumption that correct full nodes reliably respond.
The details of the fork detector is outside the scope of this paper.

\paragraph{Performance.}
It is obvious that in the case where validator set changes are rare
     (which is the case in the Cosmos Hub, the largest live network in the Cosmos ecosystem), skipping
     verification outperforms sequential verification: if the
     validator set does not change, then verifying a block on height
     1000 based on height 100 needs one step with skipping
     verification and 900 steps with sequential verification; each
     step involving expensive operations as checking hashes and
     signatures.
Still, there are several interesting performance meassurement we are
     interested in so that we are currently setting up a framework for
     experimental performance evaluation.

\bibliographystyle{alpha}
\bibliography{lit}

\newcommand{\etalchar}[1]{$^{#1}$}
\begin{thebibliography}{vGGKB{\etalchar{+}}19}

\bibitem[ABD95]{AttiyaBD95}
Hagit Attiya, Amotz Bar{-}Noy, and Danny Dolev.
\newblock Sharing memory robustly in message-passing systems.
\newblock {\em J. {ACM}}, 42(1):124--142, 1995.

\bibitem[ABSB18]{albassam2018fraud}
Mustafa Al-Bassam, Alberto Sonnino, and Vitalik Buterin.
\newblock Fraud and data availability proofs: Maximising light client security
  and scaling blockchains with dishonest majorities, 2018.

\bibitem[ADGF{\etalchar{+}}07]{ADFHW07:OPODIS}
Emmanuelle Anceaume, Carole Delporte-Gallet, Hugues Fauconnier, Michel Hurfin,
  and Josef Widder.
\newblock Clock synchronization in the {B}yzantine-recovery failure model.
\newblock In {\em OPODIS}, pages 90--104, 2007.

\bibitem[AK86]{AK86}
K.~Apt and D.~Kozen.
\newblock Limits for automatic verification of finite-state concurrent systems.
\newblock {\em IPL}, 15:307--309, 1986.

\bibitem[BCBG{\etalchar{+}}07]{BielyWCGHS07}
Martin Biely, Bernadette Charron-Bost, Antoine Gaillard, Martin Hutle,
  Andr{\'e} Schiper, and Josef Widder.
\newblock Tolerating corrupted communication.
\newblock In {\em PODC}, pages 244--253, 2007.

\bibitem[BEJQ18]{BouajjaniEJQ18}
Ahmed Bouajjani, Constantin Enea, Kailiang Ji, and Shaz Qadeer.
\newblock On the completeness of verifying message passing programs under
  bounded asynchrony.
\newblock In {\em CAV}, pages 372--391, 2018.

\bibitem[BHG87]{BernsteinHG87}
Philip~A. Bernstein, Vassos Hadzilacos, and Nathan Goodman.
\newblock {\em Concurrency Control and Recovery in Database Systems}.
\newblock Addison-Wesley, 1987.

\bibitem[BHHN00]{BarakHHN00}
Boaz Barak, Shai Halevi, Amir Herzberg, and Dalit Naor.
\newblock Clock synchronization with faults and recoveries (extended abstract).
\newblock In {\em PODC}, pages 133--142, 2000.

\bibitem[BJK{\etalchar{+}}15]{2015Bloem}
Roderick Bloem, Swen Jacobs, Ayrat Khalimov, Igor Konnov, Sasha Rubin, Helmut
  Veith, and Josef Widder.
\newblock {\em Decidability of Parameterized Verification}.
\newblock Synthesis Lectures on Distributed Computing Theory. Morgan \&
  Claypool, 2015.

\bibitem[BK16]{cosmos}
Ethan Buchman and Jae Kwon.
\newblock Cosmos whitepaper: a network of distributed ledgers, 2016.
\newblock \url{https://cosmos.network/resources/whitepaper}.

\bibitem[BKLZ19]{bunz2019flyclient}
Benedikt B{\"u}nz, Lucianna Kiffer, Loi Luu, and Mahdi Zamani.
\newblock Flyclient: Super-light clients for cryptocurrencies.
\newblock {\em IACR Cryptology ePrint Archive}, 2019:226, 2019.

\bibitem[BKM18]{bkm2018latest}
Ethan Buchman, Jae Kwon, and Zarko Milosevic.
\newblock The latest gossip on {BFT} consensus, 2018.

\bibitem[Bli20]{impl}
Blinded.
\newblock {Light Client Verification Implementation}, 2020.
\newblock URL blinded.

\bibitem[BLL{\etalchar{+}}19]{BerkovitsLLPS19}
Idan Berkovits, Marijana Lazic, Giuliano Losa, Oded Padon, and Sharon Shoham.
\newblock Verification of threshold-based distributed algorithms by
  decomposition to decidable logics.
\newblock In {\em CAV}, volume 11562 of {\em LNCS}, pages 245--266. Springer,
  2019.

\bibitem[BSA14]{BessaniSA14}
Alysson~Neves Bessani, Jo{\~{a}}o Sousa, and Eduardo Ad{\'{\i}}lio~Pelinson
  Alchieri.
\newblock State machine replication for the masses with {BFT-SMART}.
\newblock In {\em DSN}, pages 355--362, 2014.

\bibitem[BSCTV17]{ben2017scalable}
Eli Ben-Sasson, Alessandro Chiesa, Eran Tromer, and Madars Virza.
\newblock Scalable zero knowledge via cycles of elliptic curves.
\newblock {\em Algorithmica}, 79(4):1102--1160, 2017.

\bibitem[Buc16]{Buchman2016}
Ethan Buchman.
\newblock Tendermint: {Byzantine} fault tolerance in the age of {Blockchains}.
\newblock Master's thesis, University of Guelph, 2016.
\newblock \url{http://hdl.handle.net/10214/9769}.

\bibitem[But14]{subjectivity}
Vitalik Buterin.
\newblock {Proof of Stake: How I Learned to Love Weak Subjectivity}.
\newblock
  \url{https://blog.ethereum.org/2014/11/25/proof-stake-learned-love-weak-subjectivity/},
  2014.

\bibitem[But15]{pos_clients}
Vitalik Buterin.
\newblock {Light Clients and Proof of Stake}.
\newblock
  \url{https://blog.ethereum.org/2015/01/10/light-clients-proof-stake/}, 2015.

\bibitem[CBS09]{Charron-BostS09}
Bernadette Charron-Bost and Andr{\'e} Schiper.
\newblock The heard-of model: computing in distributed systems with benign
  faults.
\newblock {\em Distributed Computing}, 22(1):49--71, 2009.

\bibitem[CL02]{CastroL02}
Miguel Castro and Barbara Liskov.
\newblock Practical {B}yzantine fault tolerance and proactive recovery.
\newblock {\em {ACM} Trans. Comput. Syst.}, 20(4):398--461, 2002.

\bibitem[Cor20a]{TMBC}
Tendermint Core.
\newblock Tendermint blockchain and the rules for validating them, 2020.
\newblock
  \url{https://github.com/tendermint/spec/blob/master/spec/blockchain/blockchain.md}.

\bibitem[Cor20b]{TMCORE}
Tendermint Core.
\newblock Tendermint core, reference implementation in {Go}, 2020.
\newblock \url{https://github.com/tendermint/tendermint}.

\bibitem[CW09]{CrosbyW09}
Scott~A. Crosby and Dan~S. Wallach.
\newblock Efficient data structures for tamper-evident logging.
\newblock In {\em 18th {USENIX} Security Symposium}, pages 317--334, 2009.

\bibitem[DDMW19]{DDMW19:CAV}
Andrei Damian, Cezara Dr{\u{a}}goi, Alexandru Militaru, and Josef Widder.
\newblock {Communication-closed asynchronous protocols}.
\newblock In {\em CAV}, pages 344--363, 2019.

\bibitem[DHV{\etalchar{+}}14]{DHVWZ14}
Cezara {Dr\u{a}goi}, Thomas~A. Henzinger, Helmut Veith, Josef Widder, and
  Damien Zufferey.
\newblock {A Logic-Based Framework for Verifying Consensus Algorithms}.
\newblock In {\em VMCAI}, volume 8318 of {\em LNCS}, pages 161--181, 2014.

\bibitem[DLS88]{DLS88}
Cynthia Dwork, Nancy Lynch, and Larry Stockmeyer.
\newblock Consensus in the presence of partial synchrony.
\newblock {\em J.ACM}, 35(2):288--323, 1988.

\bibitem[EN95]{EN95}
E.A. Emerson and K.S. Namjoshi.
\newblock Reasoning about rings.
\newblock In {\em POPL}, pages 85--94, 1995.

\bibitem[GB17]{GotsmanB17}
Alexey Gotsman and Sebastian Burckhardt.
\newblock Consistency models with global operation sequencing and their
  composition.
\newblock In {\em DISC}, pages 23:1--23:16, 2017.

\bibitem[GGJ{\etalchar{+}}20]{gabizon2020plumo}
Ariel Gabizon, Kobi Gurkan, Philipp Jovanovic, Georgios Konstantopoulos, Asa
  Oines, Marek Olszewski, Michael Straka, and Eran Tromer.
\newblock Plumo: Towards scalable interoperable blockchains using ultra light
  validation systems.
\newblock 2020.

\bibitem[Gro16]{groth2016size}
Jens Groth.
\newblock On the size of pairing-based non-interactive arguments.
\newblock In {\em Annual international conference on the theory and
  applications of cryptographic techniques}, pages 305--326. Springer, 2016.

\bibitem[Iss20]{LastBlockID}
Tendermint Issues.
\newblock {\#4835} vote {\&} commitsig redundancy, 2020.
\newblock \url{https://github.com/tendermint/tendermint/issues/4835}.

\bibitem[KKT19]{KKT19}
Igor Konnov, Jure Kukovec, and Thanh{-}Hai Tran.
\newblock {TLA+} model checking made symbolic.
\newblock {\em {PACMPL}}, 3({OOPSLA}):123:1--123:30, 2019.

\bibitem[KLS16]{kiayias2016proofs}
Aggelos Kiayias, Nikolaos Lamprou, and Aikaterini-Panagiota Stouka.
\newblock Proofs of proofs of work with sublinear complexity.
\newblock In {\em International Conference on Financial Cryptography and Data
  Security}, pages 61--78. Springer, 2016.

\bibitem[KLVW17]{KLVW17:POPL}
Igor Konnov, Marijana Lazi\'{c}, Helmut Veith, and Josef Widder.
\newblock A short counterexample property for safety and liveness verification
  of fault-tolerant distributed algorithms.
\newblock In {\em POPL}, pages 719--734, 2017.

\bibitem[KMZ17]{kiayias2017non}
Aggelos Kiayias, Andrew Miller, and Dionysis Zindros.
\newblock Non-interactive proofs of proof-of-work.
\newblock {\em IACR Cryptology ePrint Archive}, 2017(963):1--42, 2017.

\bibitem[KQH18]{KraglQH18}
Bernhard Kragl, Shaz Qadeer, and Thomas~A. Henzinger.
\newblock Synchronizing the asynchronous.
\newblock In {\em CONCUR}, pages 21:1--21:17, 2018.

\bibitem[Kwo14]{kwon2014tendermint}
Jae Kwon.
\newblock Tendermint: Consensus without mining.
\newblock {\em Draft v. 0.6, fall}, 1(11), 2014.

\bibitem[Lam02]{lamport2002specifying}
Leslie Lamport.
\newblock {\em Specifying systems: The {TLA}+ language and tools for hardware
  and software engineers}.
\newblock Addison-Wesley, 2002.

\bibitem[Lam05a]{generalizedpaxos}
Leslie Lamport.
\newblock Generalized consensus and paxos.
\newblock Technical report, March 2005.

\bibitem[Lam05b]{lamport2005real-time}
Leslie Lamport.
\newblock Real-time model checking is really simple.
\newblock In {\em CHARME}, June 2005.

\bibitem[Lam11]{Lamport11a}
Leslie Lamport.
\newblock Byzantizing paxos by refinement.
\newblock In {\em DISC}, pages 211--224, 2011.

\bibitem[MS18]{meckler2018coda}
Izaak Meckler and Evan Shapiro.
\newblock Coda: Decentralized cryptocurrency at scale.
\newblock 2018.

\bibitem[Nak19]{nakamoto2019bitcoin}
Satoshi Nakamoto.
\newblock Bitcoin: A peer-to-peer electronic cash system.
\newblock Technical report, Manubot, 2019.

\bibitem[NKKJ{\etalchar{+}}17]{nikitin2017chainiac}
Kirill Nikitin, Eleftherios Kokoris-Kogias, Philipp Jovanovic, Linus Gasser,
  Nicolas Gailly, Ismail Khoffi, Justin Cappos, and Bryan Ford.
\newblock Chainiac: Proactive software-update transparency via collectively
  signed skipchains and verified builds.
\newblock In {\em USENIX Security Symposium}, pages 1271--1287, 2017.

\bibitem[NRZ{\etalchar{+}}15]{NewcombeRZMBD15}
Chris Newcombe, Tim Rath, Fan Zhang, Bogdan Munteanu, Marc Brooker, and Michael
  Deardeuff.
\newblock How amazon web services uses formal methods.
\newblock {\em Commun. {ACM}}, 58(4):66--73, 2015.

\bibitem[Ogn20]{Ogn2020}
Stevan Ognjanovic.
\newblock Verifying distributed systems with {Stainless}.
\newblock Master's thesis, EPFL, 2020.

\bibitem[PoS20]{eth_pos_faq}
{Proof of Stake FAQ}.
\newblock \url{https://github.com/ethereum/wiki/wiki/Proof-of-Stake-FAQ}, 2020.

\bibitem[Tod12]{mmrs}
Peter Todd.
\newblock {Merkle Mountain Ranges}.
\newblock
  \url{https://github.com/opentimestamps/opentimestamps-server/blob/master/doc/merkle-mountain-range.md},
  2012.

\bibitem[Unc17]{casper_tendermint}
Chjango Unchained.
\newblock {Consensus Compare: Casper vs. Tendermint}.
\newblock
  \url{https://blog.cosmos.network/consensus-compare-casper-vs-tendermint-6df154ad56ae},
  2017.

\bibitem[vGGKB{\etalchar{+}}19]{GGB19}
Klaus v.\ Gleissenthall, Rami G{\"o}khan~Kici, Alexander Bakst, Deian Stefan,
  and Ranjit Jhala.
\newblock Pretend synchrony.
\newblock In {\em POPL}, 2019.

\end{thebibliography}

 \clearpage
 \appendix
 \onecolumn

 \section{APPENDIX: Complete~\tlap{} specifications}\label{sec:tla-complete}

\newboolean{shading} 
\setboolean{shading}{false}
\makeatletter
\newlength{\symlength}
\renewcommand{\implies}{\Rightarrow}
\newcommand{\ltcolon}{\mathrel{<\!\!\mbox{:}}}
\newcommand{\colongt}{\mathrel{\!\mbox{:}\!\!>}}
\newcommand{\defeq}{\;\mathrel{\smash   
    {{\stackrel{\scriptscriptstyle\Delta}{=}}}}\;}
\newcommand{\dotdot}{\mathrel{\ldotp\ldotp}}
\newcommand{\coloncolon}{\mathrel{::\;}}
\newcommand{\eqdash}{\mathrel = \joinrel \hspace{-.28em}|}
\newcommand{\pp}{\mathbin{++}}
\newcommand{\mm}{\mathbin{--}}
\newcommand{\stst}{*\!*}
\newcommand{\slsl}{/\!/}
\newcommand{\ct}{\hat{\hspace{.4em}}}
\renewcommand{\A}{\forall}
\renewcommand{\E}{\exists}
\renewcommand{\AA}{\makebox{$\raisebox{.05em}{\makebox[0pt][l]{
   $\forall\hspace{-.517em}\forall\hspace{-.517em}\forall$}}
   \forall\hspace{-.517em}\forall \hspace{-.517em}\forall\,$}}
\newcommand{\EE}{\makebox{$\raisebox{.05em}{\makebox[0pt][l]{
   $\exists\hspace{-.517em}\exists\hspace{-.517em}\exists$}}
   \exists\hspace{-.517em}\exists\hspace{-.517em}\exists\,$}}
\newcommand{\whileop}{\.{\stackrel
  {\mbox{\raisebox{-.3em}[0pt][0pt]{$\scriptscriptstyle+\;\,$}}}
  {-\hspace{-.16em}\triangleright}}}

\newcommand{\ASSUME}{\textsc{assume }}
\newcommand{\ASSUMPTION}{\textsc{assumption }}
\newcommand{\AXIOM}{\textsc{axiom }}
\renewcommand{\BOOLEAN}{\textsc{boolean }}
\newcommand{\CASE}{\textsc{case }}
\newcommand{\CONSTANT}{\textsc{constant }}
\newcommand{\CONSTANTS}{\textsc{constants }}
\newcommand{\ELSE}{\settowidth{\symlength}{\THEN}
   \makebox[\symlength][l]{\textsc{ else}}}
\newcommand{\EXCEPT}{\textsc{ except }}
\newcommand{\EXTENDS}{\textsc{extends }}
\renewcommand{\FALSE}{\textsc{false}}
\newcommand{\IF}{\textsc{if }}
\newcommand{\IN}{\settowidth{\symlength}{\LET}
   \makebox[\symlength][l]{\textsc{in}}}
\newcommand{\INSTANCE}{\textsc{instance }}
\newcommand{\LET}{\textsc{let }}
\newcommand{\LOCAL}{\textsc{local }}
\newcommand{\MODULE}{\textsc{module }}
\newcommand{\OTHER}{\textsc{other }}
\newcommand{\STRING}{\textsc{string}}
\newcommand{\THEN}{\textsc{ then }}
\newcommand{\THEOREM}{\textsc{theorem }}
\newcommand{\LEMMA}{\textsc{lemma }}
\newcommand{\PROPOSITION}{\textsc{proposition }}
\newcommand{\COROLLARY}{\textsc{corollary }}
\renewcommand{\TRUE}{\textsc{true}}
\newcommand{\VARIABLE}{\textsc{variable }}
\newcommand{\VARIABLES}{\textsc{variables }}
\newcommand{\WITH}{\textsc{ with }}
\newcommand{\WF}{\textrm{WF}}
\newcommand{\SF}{\textrm{SF}}
\newcommand{\CHOOSE}{\textsc{choose }}
\newcommand{\ENABLED}{\textsc{enabled }}
\newcommand{\UNCHANGED}{\textsc{unchanged }}
\renewcommand{\SUBSET}{\textsc{subset }}
\renewcommand{\UNION}{\textsc{union }}
\renewcommand{\DOMAIN}{\textsc{domain }}
\newcommand{\BY}{\textsc{by }}
\newcommand{\OBVIOUS}{\textsc{obvious }}
\newcommand{\HAVE}{\textsc{have }}
\newcommand{\QED}{\textsc{qed }}
\newcommand{\TAKE}{\textsc{take }}
\newcommand{\DEF}{\textsc{ def }}
\newcommand{\HIDE}{\textsc{hide }}
\newcommand{\RECURSIVE}{\textsc{recursive }}
\newcommand{\USE}{\textsc{use }}
\newcommand{\DEFINE}{\textsc{define }}
\newcommand{\PROOF}{\textsc{proof }}
\newcommand{\WITNESS}{\textsc{witness }}
\newcommand{\PICK}{\textsc{pick }}
\newcommand{\DEFS}{\textsc{defs }}
\newcommand{\PROVE}{\settowidth{\symlength}{\ASSUME}
   \makebox[\symlength][l]{\textsc{prove}}\@s{-4.1}}
\newcommand{\SUFFICES}{\textsc{suffices }}
\newcommand{\NEW}{\textsc{new }}
\newcommand{\LAMBDA}{\textsc{lambda }}
\newcommand{\STATE}{\textsc{state }}
\newcommand{\ACTION}{\textsc{action }}
\newcommand{\TEMPORAL}{\textsc{temporal }}
\newcommand{\ONLY}{\textsc{only }}              
\newcommand{\OMITTED}{\textsc{omitted }}        
\newcommand{\@pfstepnum}[2]{\ensuremath{\langle#1\rangle}\textrm{#2}}
\newcommand{\bang}{\@s{1}\mbox{\small !}\@s{1}}
\newcommand{\p@barbar}{\ifpcalsymbols
   \,\,\rule[-.25em]{.075em}{1em}\hspace*{.2em}\rule[-.25em]{.075em}{1em}\,\,
   \else \,||\,\fi}
\newcommand{\p@fair}{\textbf{fair }}
\newcommand{\p@semicolon}{\textbf{\,; }}
\newcommand{\p@algorithm}{\textbf{algorithm }}
\newcommand{\p@mmfair}{\textbf{-{}-fair }}
\newcommand{\p@mmalgorithm}{\textbf{-{}-algorithm }}
\newcommand{\p@assert}{\textbf{assert }}
\newcommand{\p@await}{\textbf{await }}
\newcommand{\p@begin}{\textbf{begin }}
\newcommand{\p@end}{\textbf{end }}
\newcommand{\p@call}{\textbf{call }}
\newcommand{\p@define}{\textbf{define }}
\newcommand{\p@do}{\textbf{ do }}
\newcommand{\p@either}{\textbf{either }}
\newcommand{\p@or}{\textbf{or }}
\newcommand{\p@goto}{\textbf{goto }}
\newcommand{\p@if}{\textbf{if }}
\newcommand{\p@then}{\,\,\textbf{then }}
\newcommand{\p@else}{\ifcsyntax \textbf{else } \else \,\,\textbf{else }\fi}
\newcommand{\p@elsif}{\,\,\textbf{elsif }}
\newcommand{\p@macro}{\textbf{macro }}
\newcommand{\p@print}{\textbf{print }}
\newcommand{\p@procedure}{\textbf{procedure }}
\newcommand{\p@process}{\textbf{process }}
\newcommand{\p@return}{\textbf{return}}
\newcommand{\p@skip}{\textbf{skip}}
\newcommand{\p@variable}{\textbf{variable }}
\newcommand{\p@variables}{\textbf{variables }}
\newcommand{\p@while}{\textbf{while }}
\newcommand{\p@when}{\textbf{when }}
\newcommand{\p@with}{\textbf{with }}
\newcommand{\p@lparen}{\textbf{(\,\,}}
\newcommand{\p@rparen}{\textbf{\,\,) }}   
\newcommand{\p@lbrace}{\textbf{\{\,\,}}   
\newcommand{\p@rbrace}{\textbf{\,\,\} }}

\renewcommand{\_}{\rule{.4em}{.06em}\hspace{.05em}}
\newlength{\equalswidth}
\let\oldin=\in
\let\oldnotin=\notin
\renewcommand{\in}{
   {\settowidth{\equalswidth}{$\.{=}$}\makebox[\equalswidth][c]{$\oldin$}}}
\renewcommand{\notin}{
   {\settowidth{\equalswidth}{$\.{=}$}\makebox[\equalswidth]{$\oldnotin$}}}

\newlength{\charwidth}\settowidth{\charwidth}{{\small\tt M}}
\newlength{\boxrulewd}\setlength{\boxrulewd}{.4pt}
\newlength{\boxlineht}\setlength{\boxlineht}{.5\baselineskip}
\newcommand{\boxsep}{\charwidth}
\newlength{\boxruleht}\setlength{\boxruleht}{.5ex}
\newlength{\boxruledp}\setlength{\boxruledp}{-\boxruleht}
\addtolength{\boxruledp}{\boxrulewd}
\newcommand{\boxrule}{\leaders\hrule height \boxruleht depth \boxruledp
                      \hfill\mbox{}}
\newcommand{\@computerule}{
  \setlength{\boxruleht}{.5ex}
  \setlength{\boxruledp}{-\boxruleht}
  \addtolength{\boxruledp}{\boxrulewd}}

\newcommand{\bottombar}{\hspace{-\boxsep}
  \raisebox{-\boxrulewd}[0pt][0pt]{\rule[.5ex]{\boxrulewd}{\boxlineht}}
  \boxrule
  \raisebox{-\boxrulewd}[0pt][0pt]{
      \rule[.5ex]{\boxrulewd}{\boxlineht}}\hspace{-\boxsep}\vspace{0pt}}

\newcommand{\moduleLeftDash}
   {\hspace*{-\boxsep}
     \raisebox{-\boxlineht}[0pt][0pt]{\rule[.5ex]{\boxrulewd
               }{\boxlineht}}
    \boxrule\hspace*{.4em }}

\newcommand{\moduleRightDash}
    {\hspace*{.4em}\boxrule
    \raisebox{-\boxlineht}[0pt][0pt]{\rule[.5ex]{\boxrulewd
               }{\boxlineht}}\hspace{-\boxsep}}

\newcommand{\midbar}{\hspace{-\boxsep}\raisebox{-.5\boxlineht}[0pt][0pt]{
   \rule[.5ex]{\boxrulewd}{\boxlineht}}\boxrule\raisebox{-.5\boxlineht
   }[0pt][0pt]{\rule[.5ex]{\boxrulewd}{\boxlineht}}\hspace{-\boxsep}}

\newif\ifpcalshading \pcalshadingfalse
\newif\ifpcalsymbols \pcalsymbolsfalse
\newif\ifcsyntax     \csyntaxtrue

\newlength{\pcalvspace}\setlength{\pcalvspace}{0pt}
\newcommand{\@pvspace}[1]{
  \ifpcalshading
     \par\global\setlength{\pcalvspace}{#1}
  \else
     \par\vspace{#1}
  \fi
}

\newlength{\lcomindent}
\setlength{\lcomindent}{0pt}

\newcommand\tstrut
  {\raisebox{\vshadelen}{\raisebox{-.25em}{\rule{0pt}{1.15em}}}
   \global\setlength{\vshadelen}{0pt}}
\newcommand\rstrut{\raisebox{-.25em}{\rule{0pt}{1.15em}}
 \global\setlength{\vshadelen}{0pt}}

\renewcommand{\.}[1]{\ensuremath{\mbox{}#1\mbox{}}}

\newcommand{\@s}[1]{\hspace{#1pt}}           

\newlength{\@xlen}
\newcommand\xtstrut
  {\setlength{\@xlen}{1.05em}
   \addtolength{\@xlen}{\pcalvspace}
    \raisebox{\vshadelen}{\raisebox{-.25em}{\rule{0pt}{\@xlen}}}
   \global\setlength{\vshadelen}{0pt}
   \global\setlength{\pcalvspace}{0pt}}

\newcommand{\@x}[1]{\par
  \ifpcalshading
  \makebox[0pt][l]{\shadebox{\xtstrut\hspace*{\textwidth}}}
  \fi
  \mbox{$\mbox{}#1\mbox{}$}}  

\newcommand{\@xx}[1]{\mbox{$\mbox{}#1\mbox{}$}}  

\newcommand{\@y}[1]{\mbox{\footnotesize\hspace{.65em}
  \ifthenelse{\boolean{shading}}{
      \shadebox{#1\hspace{-\the\lastskip}\rstrut}}
               {#1\hspace{-\the\lastskip}\rstrut}}}

\newcommand{\@z}[1]{\makebox[0pt][l]{\footnotesize
  \ifthenelse{\boolean{shading}}{
      \shadebox{#1\hspace{-\the\lastskip}\rstrut}}
               {#1\hspace{-\the\lastskip}\rstrut}}}

\newcommand{\@w}[1]{\textsf{``{#1}''}}

\def\graymargin{1}

\newlength{\templena}
\newlength{\templenb}
\newsavebox{\tempboxa}
\newcommand{\shadebox}[1]{{\setlength{\fboxsep}{\graymargin pt}
     \savebox{\tempboxa}{#1}
     \settoheight{\templena}{\usebox{\tempboxa}}
     \settodepth{\templenb}{\usebox{\tempboxa}}
     \hspace*{-\fboxsep}\raisebox{0pt}[\templena][\templenb]
        {\colorbox{boxshade}{\usebox{\tempboxa}}}\hspace*{-\fboxsep}}}

\newlength{\vshadelen}
\setlength{\vshadelen}{0pt}
\newcommand{\vshade}[1]{\ifthenelse{\boolean{shading}}
   {\global\setlength{\vshadelen}{#1pt}}
   {\vspace{#1pt}}}

\newlength{\boxwidth}
\newlength{\multicommentdepth}

\newcounter{pardepth}
\setcounter{pardepth}{0}

\newcommand{\setgmargin}[1]{
  \expandafter\xdef\csname gmargin\roman{pardepth}\endcsname{#1}}
\newcommand{\thegmargin}{\csname gmargin\roman{pardepth}\endcsname}
\newcommand{\gmargin}{0pt}

\newsavebox{\tempsbox}

\newlength{\@cparht}
\newlength{\@cpardp}
\newenvironment{cpar}[6]{
  \addtocounter{pardepth}{-#1}
  \ifthenelse{\boolean{shading}}{\par\begin{lrbox}{\tempsbox}
                                 \begin{minipage}[t]{\linewidth}}{}
  \begin{list}{}{
     \edef\temp{\thegmargin}
     \ifthenelse{\equal{#3}{T}}
       {\settowidth{\leftmargin}{\hspace{\temp}\footnotesize #6\hspace{#5pt}}
        \addtolength{\leftmargin}{#4pt}}
       {\setlength{\leftmargin}{#4pt}
        \addtolength{\leftmargin}{#5pt}
        \addtolength{\leftmargin}{\temp}
        \setlength{\itemindent}{-#5pt}}
      \ifthenelse{\equal{#2}{T}}{\addtocounter{pardepth}{1}
                                 \setgmargin{\the\leftmargin}}{}
      \setlength{\labelwidth}{0pt}
      \setlength{\labelsep}{0pt}
      \setlength{\itemindent}{-\leftmargin}
      \setlength{\topsep}{0pt}
      \setlength{\parsep}{0pt}
      \setlength{\partopsep}{0pt}
      \setlength{\parskip}{0pt}
      \setlength{\itemsep}{0pt}
      \setlength{\itemindent}{#4pt}
      \addtolength{\itemindent}{-\leftmargin}}
   \ifthenelse{\equal{#3}{T}}
      {\item[\tstrut\footnotesize \hspace{\temp}{#6}\hspace{#5pt}]
        }
      {\item[\tstrut\hspace{\temp}]
         }
   \footnotesize}
 {\hspace{-\the\lastskip}\tstrut
 \end{list}
  \ifthenelse{\boolean{shading}}
          {\end{minipage}
           \end{lrbox}
           \ifpcalshading
             \setlength{\@cparht}{\ht\tempsbox}
             \setlength{\@cpardp}{\dp\tempsbox}
             \addtolength{\@cparht}{.15em}
             \addtolength{\@cpardp}{.2em}
             \addtolength{\@cparht}{\@cpardp}
            \addtolength{\@cparht}{\pcalvspace}
             \ifdim \pcalvspace > .8em
               \addtolength{\pcalvspace}{-.2em}
               \hspace*{-\lcomindent}
               \shadebox{\rule{0pt}{\pcalvspace}\hspace*{\textwidth}}\par
               \global\setlength{\pcalvspace}{0pt}
               \fi
             \hspace*{-\lcomindent}
             \makebox[0pt][l]{\raisebox{-\@cpardp}[0pt][0pt]{
                 \shadebox{\rule{0pt}{\@cparht}\hspace*{\textwidth}}}}
             \hspace*{\lcomindent}\usebox{\tempsbox}
             \par
           \else
             \shadebox{\usebox{\tempsbox}}\par
           \fi}
           {}
  }

 \newcommand{\xtest}[1]{\par
 \makebox[0pt][l]{\shadebox{\xtstrut\hspace*{\textwidth}}}
 \mbox{$\mbox{}#1\mbox{}$}}

\newenvironment{lcom}[1]{
  \setlength{\lcomindent}{#1pt} 
  \par\vspace{.2em}
  \sloppypar
  \setcounter{pardepth}{0}
  \footnotesize
  \begin{list}{}{
    \setlength{\leftmargin}{#1pt}
    \setlength{\labelwidth}{0pt}
    \setlength{\labelsep}{0pt}
    \setlength{\itemindent}{0pt}
    \setlength{\topsep}{0pt}
    \setlength{\parsep}{0pt}
    \setlength{\partopsep}{0pt}
    \setlength{\parskip}{0pt}}
    \item[]}
  {\end{list}\vspace{.3em}\setlength{\lcomindent}{0pt}
 }

\newlength{\xmcomlen}
\newenvironment{mcom}[1]{
  \setcounter{pardepth}{0}
  \hspace{.65em}
  \begin{lrbox}{\alignbox}\sloppypar
      \setboolean{shading}{false}
      \setlength{\boxwidth}{#1pt}
      \addtolength{\boxwidth}{-.65em}
      \begin{minipage}[t]{\boxwidth}\footnotesize
      \parskip=0pt\relax}
       {\end{minipage}\end{lrbox}
       \setlength{\xmcomlen}{\textwidth}
       \addtolength{\xmcomlen}{-\wd\alignbox}
       \settodepth{\alignwidth}{\usebox{\alignbox}}
       \global\setlength{\multicommentdepth}{\alignwidth}
       \setlength{\boxwidth}{\alignwidth}
       \global\addtolength{\alignwidth}{-\maxdepth}
       \addtolength{\boxwidth}{.1em}
      \raisebox{0pt}[0pt][0pt]{
        \ifthenelse{\boolean{shading}}
          {\ifpcalshading
             \hspace*{-\xmcomlen}
             \shadebox{\rule[-\boxwidth]{0pt}{0pt}\hspace*{\xmcomlen}
                          \usebox{\alignbox}}
           \else
             \shadebox{\usebox{\alignbox}}
           \fi
          }
          {\usebox{\alignbox}}}
       \vspace*{\alignwidth}\pagebreak[0]\vspace{-\alignwidth}\par}

\newcommand{\multivspace}[1]{\addtolength{\multicommentdepth}{-#1\baselineskip}
 \addtolength{\multicommentdepth}{1.2em}
 \ifthenelse{\lengthtest{\multicommentdepth > 0pt}}
    {\par\vspace{\multicommentdepth}\par}{}}

\newcommand{\raisedDash}[3]{\raisebox{#2ex}{\setlength{\alignwidth}{.5em}
  \rule{#1\alignwidth}{#3em}}}

\newcommand{\cdash}[1]{\raisedDash{#1}{.5}{.04}}
\newcommand{\usdash}[1]{\raisedDash{#1}{0}{.04}}
\newcommand{\ceqdash}[1]{\raisedDash{#1}{.5}{.08}}
\newcommand{\tdash}[1]{\raisedDash{#1}{1}{.08}}

\newlength{\spacewidth}
\setlength{\spacewidth}{.2em}
\newcommand{\e}[1]{\hspace{#1\spacewidth}}

\newlength{\alignboxwidth}
\newlength{\alignwidth}
\newsavebox{\alignbox}

\newcommand{\al}[3]{
  \typeout{\%{#1}{#2}{\the\alignwidth}}
  \cl{#3}}

\newcommand{\cl}[1]{
  \savebox{\alignbox}{\mbox{$\mbox{}#1\mbox{}$}}
  \settowidth{\alignboxwidth}{\usebox{\alignbox}}
  \addtolength{\alignwidth}{\alignboxwidth}
  \usebox{\alignbox}}

\newcommand{\fl}[1]{
  \par
  \savebox{\alignbox}{\mbox{$\mbox{}#1\mbox{}$}}
  \settowidth{\alignwidth}{\usebox{\alignbox}}
  \usebox{\alignbox}}

\makeatletter
\newcommand{\tlx@c}{\c@tlx@ctr\advance\c@tlx@ctr\@ne}
\newcounter{tlx@ctr}
\c@tlx@ctr=\itfam \multiply\c@tlx@ctr"100\relax \advance\c@tlx@ctr "7061\relax
\mathcode`a=\tlx@c \mathcode`b=\tlx@c \mathcode`c=\tlx@c \mathcode`d=\tlx@c
\mathcode`e=\tlx@c \mathcode`f=\tlx@c \mathcode`g=\tlx@c \mathcode`h=\tlx@c
\mathcode`i=\tlx@c \mathcode`j=\tlx@c \mathcode`k=\tlx@c \mathcode`l=\tlx@c
\mathcode`m=\tlx@c \mathcode`n=\tlx@c \mathcode`o=\tlx@c \mathcode`p=\tlx@c
\mathcode`q=\tlx@c \mathcode`r=\tlx@c \mathcode`s=\tlx@c \mathcode`t=\tlx@c
\mathcode`u=\tlx@c \mathcode`v=\tlx@c \mathcode`w=\tlx@c \mathcode`x=\tlx@c
\mathcode`y=\tlx@c \mathcode`z=\tlx@c
\c@tlx@ctr=\itfam \multiply\c@tlx@ctr"100\relax \advance\c@tlx@ctr "7041\relax
\mathcode`A=\tlx@c \mathcode`B=\tlx@c \mathcode`C=\tlx@c \mathcode`D=\tlx@c
\mathcode`E=\tlx@c \mathcode`F=\tlx@c \mathcode`G=\tlx@c \mathcode`H=\tlx@c
\mathcode`I=\tlx@c \mathcode`J=\tlx@c \mathcode`K=\tlx@c \mathcode`L=\tlx@c
\mathcode`M=\tlx@c \mathcode`N=\tlx@c \mathcode`O=\tlx@c \mathcode`P=\tlx@c
\mathcode`Q=\tlx@c \mathcode`R=\tlx@c \mathcode`S=\tlx@c \mathcode`T=\tlx@c
\mathcode`U=\tlx@c \mathcode`V=\tlx@c \mathcode`W=\tlx@c \mathcode`X=\tlx@c
\mathcode`Y=\tlx@c \mathcode`Z=\tlx@c
\makeatother

\newenvironment{describe}[1]
   {\begin{list}{}{\settowidth{\labelwidth}{#1}
            \setlength{\labelsep}{.5em}
            \setlength{\leftmargin}{\labelwidth}
            \addtolength{\leftmargin}{\labelsep}
            \addtolength{\leftmargin}{\parindent}
            \def\makelabel##1{\rm ##1\hfill}}
            \setlength{\topsep}{0pt}}
   {\end{list}}

\makeatletter
\def\tla{\let\%\relax
         \@bsphack
         \typeout{\%{\the\linewidth}}
             \let\do\@makeother\dospecials\catcode`\^^M\active
             \let\verbatim@startline\relax
             \let\verbatim@addtoline\@gobble
             \let\verbatim@processline\relax
             \let\verbatim@finish\relax
             \verbatim@}
\let\endtla=\@esphack

\let\pcal=\tla
\let\endpcal=\endtla
\let\ppcal=\tla
\let\endppcal=\endtla

\newenvironment{tlatex}{\@computerule
                        \setlength{\parindent}{0pt}
                       \makeatletter\chardef\%=`\%}{}

\def\notla{\let\%\relax
         \@bsphack
             \let\do\@makeother\dospecials\catcode`\^^M\active
             \let\verbatim@startline\relax
             \let\verbatim@addtoline\@gobble
             \let\verbatim@processline\relax
             \let\verbatim@finish\relax
             \verbatim@}
\let\endnotla=\@esphack

\let\nopcal=\notla
\let\endnopcal=\endnotla
\let\noppcal=\notla
\let\endnoppcal=\endnotla

\tlatex
\setboolean{shading}{true}
\@x{}\moduleLeftDash\@xx{ {\MODULE} Blockchain\_A\_1}\moduleRightDash\@xx{}
\begin{lcom}{5.0}
\begin{cpar}{0}{F}{F}{0}{0}{}
 This is a high-level specification of \ensuremath{Tendermint}
 \ensuremath{blockchain
} that is designed specifically for the light client.
 Validators have the voting power of one. If you like to model various
 voting powers, introduce multiple copies of the same validator
 (do not forget to give them unique names though).
\end{cpar}
\end{lcom}
\@x{ {\EXTENDS} Integers ,\, FiniteSets}
\@pvspace{8.0pt}
\@x{ Min ( a ,\, b ) \.{\defeq} {\IF} a \.{<} b \.{\THEN} a \.{\ELSE} b}
\@pvspace{8.0pt}
\@x{ {\CONSTANT}}
\@x{\@s{8.2} AllNodes ,\,}
\@x{\@s{16.4}}
\@y{\@s{0}
 a set of all nodes that can act as validators (correct and faulty)
}
\@xx{}
\@x{\@s{8.2} ULTIMATE\_HEIGHT ,\,}
\@x{\@s{16.4}}
\@y{\@s{0}
 a maximal height that can be ever reached (modelling artifact)
}
\@xx{}
\@x{\@s{8.2} TRUSTING\_PERIOD}
\@x{\@s{16.4}}
\@y{\@s{0}
 the period within which the validators are trusted
}
\@xx{}
\@pvspace{8.0pt}
\@x{ Heights \.{\defeq} 1 \.{\dotdot} ULTIMATE\_HEIGHT\@s{8.2}}
\@y{\@s{0}
 possible heights
}
\@xx{}
\@pvspace{8.0pt}
\@x{}
\@y{\@s{0}
 A commit is just a set of nodes who have committed the block
}
\@xx{}
\@x{ Commits \.{\defeq} {\SUBSET} AllNodes}
\@pvspace{8.0pt}
\begin{lcom}{7.5}
\begin{cpar}{0}{F}{F}{0}{0}{}
The set of all block headers that can be on the \ensuremath{blockchain}.
 This is a simplified version of the Block data structure in the actual
 implementation.
\end{cpar}
\end{lcom}
\@x{ BlockHeaders \.{\defeq} [}
\@x{\@s{8.2} height \.{:} Heights ,\,}
\@x{\@s{16.4}}
\@y{\@s{0}
 the block height
}
\@xx{}
\@x{\@s{8.2} time \.{:} Int ,\,}
\@x{\@s{16.4}}
\@y{\@s{0}
 the block timestamp in some integer units
}
\@xx{}
\@x{\@s{8.2} lastCommit \.{:} Commits ,\,}
\@x{\@s{16.4}}
\@y{\@s{0}
 the nodes who have voted on the previous block, the set itself instead of a
 hash
}
\@xx{}
\begin{lcom}{12.5}
\begin{cpar}{0}{F}{F}{0}{0}{}
 in the implementation, only the hashes of \ensuremath{V} and
 \ensuremath{NextV} are stored in a block,
 as \ensuremath{V} and \ensuremath{NextV} are stored in the application state
\end{cpar}
\end{lcom}
\@x{\@s{8.2} VS \.{:} {\SUBSET} AllNodes ,\,}
\@x{\@s{22.45}}
\@y{\@s{0}
 the validators of this bloc. We store the validators instead of the hash.
}
\@xx{}
\@x{\@s{8.2} NextVS \.{:} {\SUBSET} AllNodes}
\@x{\@s{16.4}}
\@y{\@s{0}
 the validators of the next block. We store the next validators instead of
 the hash.
}
\@xx{}
\@x{ ]}
\@pvspace{8.0pt}
\@x{}
\@y{\@s{0}
 A signed header is just a header together with a set of commits
}
\@xx{}
 \@x{ LightBlocks \.{\defeq} [ header \.{:} BlockHeaders ,\, Commits \.{:}
 Commits ]}
\@pvspace{8.0pt}
\@x{ {\VARIABLES}}
\@x{\@s{16.4} now ,\,}
\@x{\@s{38.95}}
\@y{\@s{0}
 the current global time in integer units
}
\@xx{}
\@x{\@s{16.4} blockchain ,\,}
\@x{\@s{16.4}}
\@y{\@s{0}
 A sequence of \ensuremath{BlockHeaders}, which gives us a bird view of the
 \ensuremath{blockchain}.
}
\@xx{}
\@x{\@s{16.4} Faulty}
\begin{lcom}{17.5}
\begin{cpar}{0}{F}{F}{0}{0}{}
A set of faulty nodes, which can act as validators. We assume that the set
 of faulty processes is non-decreasing. If a process has recovered, it should
 connect using a different id.
\end{cpar}
\end{lcom}
\@x{}
\@y{\@s{0}
 all variables, to be used with \ensuremath{{\UNCHANGED}
}}
\@xx{}
\@x{ vars \.{\defeq} {\langle} now ,\, blockchain ,\, Faulty {\rangle}}
\@pvspace{8.0pt}
\@x{}
\@y{\@s{0}
 The set of all correct nodes in a state
}
\@xx{}
\@x{ Corr \.{\defeq} AllNodes \.{\,\backslash\,} Faulty}
\@pvspace{8.0pt}
\@x{}
\@y{\@s{0}
 \ensuremath{APALACHE} annotations
}
\@xx{}
\@x{ a \.{\ltcolon} b \.{\defeq} a}
\@y{\@s{0}
 type annotation
}
\@xx{}
\@pvspace{8.0pt}
\@x{ NT \.{\defeq} {\STRING}}
\@x{ NodeSet ( S ) \.{\defeq} S \.{\ltcolon} \{ NT \}}
\@x{ EmptyNodeSet \.{\defeq} NodeSet ( \{ \} )}
\@pvspace{8.0pt}
 \@x{ BT \.{\defeq} [ height \.{\mapsto} Int ,\, time \.{\mapsto} Int ,\,
 lastCommit \.{\mapsto} \{ NT \} ,\, VS \.{\mapsto} \{ NT \} ,\, NextVS
 \.{\mapsto} \{ NT \} ]}
\@pvspace{8.0pt}
 \@x{ LBT \.{\defeq} [ header \.{\mapsto} BT ,\, Commits \.{\mapsto} \{ NT \}
 ]}
\@x{}
\@y{\@s{0}
 end of \ensuremath{APALACHE} annotations
}
\@xx{}
\@pvspace{8.0pt}
\@x{}
\@y{
 ***************************** \ensuremath{BLOCKCHAIN}
 ***********************************
}
\@xx{}
\@pvspace{8.0pt}
\@x{}
\@y{\@s{0}
 the header is still within the trusting period
}
\@xx{}
\@x{ InTrustingPeriod ( header ) \.{\defeq}}
\@x{\@s{16.4} now \.{\leq} header . time \.{+} TRUSTING\_PERIOD}
\@pvspace{8.0pt}
\begin{lcom}{2.5}
\begin{cpar}{0}{F}{F}{0}{0}{}
 Given a function \ensuremath{pVotingPower \.{\in} D \.{\rightarrow} Powers}
 for some \ensuremath{D \.{\subseteq} AllNodes
 } and \ensuremath{pNodes \.{\subseteq} D}, test whether the set
 \ensuremath{pNodes \.{\subseteq} AllNodes} has
 more than 2/3 of voting power among the nodes in \ensuremath{D}.
\end{cpar}
\end{lcom}
\@x{ TwoThirds ( pVS ,\, pNodes ) \.{\defeq}}
\@x{\@s{16.4} \.{\LET} TP \.{\defeq} Cardinality ( pVS )}
\@x{\@s{36.79} SP\@s{1.53} \.{\defeq} Cardinality ( pVS \.{\cap} pNodes )}
\@x{\@s{16.4} \.{\IN}}
\@x{\@s{16.4} 3 \.{*} SP \.{>} 2 \.{*} TP}
\@y{\@s{0}
 when thinking in real numbers, not integers: \ensuremath{SP \.{>} 2}.0 /
 3.\ensuremath{0 \.{*} TP
}}
\@xx{}
\@pvspace{8.0pt}
\begin{lcom}{2.5}
\begin{cpar}{0}{F}{F}{0}{0}{}
 Given a set of \ensuremath{FaultyNodes}, test whether the voting power of the
 correct nodes in \ensuremath{D
} is more than 2/3 of the voting power of the faulty nodes in \ensuremath{D}.
\end{cpar}
\end{lcom}
\@x{ IsCorrectPower ( pFaultyNodes ,\, pVS ) \.{\defeq}}
 \@x{\@s{16.4} \.{\LET} FN\@s{0.62} \.{\defeq} pFaultyNodes \.{\cap}
 pVS\@s{8.2}}
\@y{\@s{0}
 faulty nodes in \ensuremath{pNodes
}}
\@xx{}
\@x{\@s{36.79} CN \.{\defeq} pVS \.{\,\backslash\,} pFaultyNodes\@s{45.1}}
\@y{\@s{0}
 correct nodes in \ensuremath{pNodes
}}
\@xx{}
\@x{\@s{36.79} CP\@s{1.26} \.{\defeq} Cardinality ( CN )\@s{57.4}}
\@y{\@s{0}
 power of the correct nodes
}
\@xx{}
\@x{\@s{36.79} FP\@s{1.89} \.{\defeq} Cardinality ( FN )\@s{24.46}}
\@y{\@s{0}
 power of the faulty nodes
}
\@xx{}
\@x{\@s{16.4} \.{\IN}}
\@x{\@s{16.4}}
\@y{\@s{0}
 \ensuremath{CP \.{+} FP \.{=} TP} is the total voting power, so we write
 \ensuremath{CP \.{>} 2}.0 / \ensuremath{3 \.{*} TP} as follows:
}
\@xx{}
\@x{\@s{16.4} CP \.{>} 2 \.{*} FP}
\@y{\@s{0}
 Note: when \ensuremath{FP \.{=} 0}, this implies \ensuremath{CP \.{>} 0}.
}
\@xx{}
\@pvspace{8.0pt}
\@x{}
\@y{\@s{0}
 This is what we believe is the assumption about failures in
 \ensuremath{Tendermint
}}
\@xx{}
\@x{ FaultAssumption ( pFaultyNodes ,\, pNow ,\, pBlockchain ) \.{\defeq}}
\@x{\@s{16.4} \A\, h \.{\in} Heights \.{:}}
 \@x{\@s{24.59} pBlockchain [ h ] . time \.{+} TRUSTING\_PERIOD \.{>} pNow
 \.{\implies}}
\@x{\@s{32.8} IsCorrectPower ( pFaultyNodes ,\, pBlockchain [ h ] . NextVS )}
\@pvspace{8.0pt}
\@x{}
\@y{\@s{0}
 Can a block be produced by a correct peer, or an authenticated
 \ensuremath{Byzantine} peer
}
\@xx{}
\@x{ IsLightBlockAllowedByDigitalSignatures ( ht ,\, block ) \.{\defeq}}
\@x{\@s{16.4} \.{\lor} block . header \.{=} blockchain [ ht ]}
\@y{\@s{0}
 signed by correct and faulty (maybe)
}
\@xx{}
 \@x{\@s{16.4} \.{\lor} block . Commits \.{\subseteq} Faulty \.{\land} block .
 header . height \.{=} ht}
\@y{\@s{0}
 signed only by faulty
}
\@xx{}
\@pvspace{8.0pt}
\begin{lcom}{2.5}
\begin{cpar}{0}{F}{F}{0}{0}{}
 Initialize the \ensuremath{blockchain} to the ultimate height right in the
 initial states.
 We pick the faulty validators statically, but that should not affect the
 light client.
\end{cpar}
\end{lcom}
\@x{ InitToHeight \.{\defeq}}
\@x{\@s{8.2} \.{\land} Faulty \.{\in} {\SUBSET} AllNodes}
\@y{\@s{0}
 some nodes may fail
}
\@xx{}
\@x{\@s{8.2}}
\@y{\@s{0}
 pick the validator sets and last commits
}
\@xx{}
 \@x{\@s{8.2} \.{\land} \E\, vs ,\, lastCommit \.{\in} [ Heights
 \.{\rightarrow} {\SUBSET} AllNodes ] \.{:}}
\@x{\@s{19.31} \E\, timestamp \.{\in} [ Heights \.{\rightarrow} Int ] \.{:}}
\@x{\@s{26.53}}
\@y{\@s{0}
 now is at least as early as the timestamp in the last block
}
\@xx{}
 \@x{\@s{26.53} \.{\land} \E\, tm \.{\in} Int \.{:} now \.{=} tm \.{\land} tm
 \.{\geq} timestamp [ ULTIMATE\_HEIGHT ]}
\@x{\@s{26.53}}
\@y{\@s{0}
 the genesis starts on day 1
}
\@xx{}
\@x{\@s{26.53} \.{\land} timestamp [ 1 ] \.{=} 1}
\@x{\@s{26.53} \.{\land} vs [ 1 ] \.{=} AllNodes}
\@x{\@s{26.53} \.{\land} lastCommit [ 1 ] \.{=} EmptyNodeSet}
 \@x{\@s{26.53} \.{\land} \A\, h \.{\in} Heights \.{\,\backslash\,} \{ 1 \}
 \.{:}}
 \@x{\@s{34.73} \.{\land} lastCommit [ h ] \.{\subseteq} vs [ h \.{-} 1
 ]\@s{48.52}}
\@y{\@s{0}
 the non-validators cannot commit
}
\@xx{}
\@x{\@s{34.73} \.{\land} TwoThirds ( vs [ h \.{-} 1 ] ,\, lastCommit [ h ] )}
\@y{\@s{0}
 the commit has \ensuremath{\.{>}2}/3 of validator votes
}
\@xx{}
\@x{\@s{34.73} \.{\land} IsCorrectPower ( Faulty ,\, vs [ h ] )\@s{31.23}}
\@y{\@s{0}
 the correct validators have \ensuremath{\.{>}2}/3 of power
}
\@xx{}
 \@x{\@s{34.73} \.{\land} timestamp [ h ] \.{>} timestamp [ h \.{-} 1
 ]\@s{17.85}}
\@y{\@s{0}
 the time grows monotonically
}
\@xx{}
 \@x{\@s{34.73} \.{\land} timestamp [ h ] \.{<} timestamp [ h \.{-} 1 ] \.{+}
 TRUSTING\_PERIOD\@s{12.29}}
\@y{\@s{0}
 but not too fast
}
\@xx{}
\@x{\@s{26.53}}
\@y{\@s{0}
 form the block chain out of validator sets and commits (this makes apalache
 faster)
}
\@xx{}
\@x{\@s{26.53} \.{\land} blockchain \.{=} [ h \.{\in} Heights \.{\mapsto}}
\@x{\@s{45.84} [ height \.{\mapsto} h ,\,}
\@x{\@s{48.62} time \.{\mapsto} timestamp [ h ] ,\,}
\@x{\@s{48.62} VS \.{\mapsto} vs [ h ] ,\,}
 \@x{\@s{48.62} NextVS \.{\mapsto} {\IF} h \.{<} ULTIMATE\_HEIGHT \.{\THEN} vs
 [ h \.{+} 1 ] \.{\ELSE} AllNodes ,\,}
\@x{\@s{48.62} lastCommit \.{\mapsto} lastCommit [ h ] ]}
\@x{\@s{45.84} ]}
\@pvspace{16.0pt}
\@x{}
\@y{\@s{0}
 is the \ensuremath{blockchain} in the faulty zone where the
 \ensuremath{Tendermint} security model does not apply
}
\@xx{}
\@x{ InFaultyZone \.{\defeq}}
\@x{\@s{8.2} {\lnot} FaultAssumption ( Faulty ,\, now ,\, blockchain )}
\@pvspace{8.0pt}
\@x{}
\@y{
 ******************** \ensuremath{BLOCKCHAIN} ACTIONS
 *******************************
}
\@xx{}
\begin{lcom}{5.0}
\begin{cpar}{0}{F}{F}{0}{0}{}
Advance the clock by zero or more time units.
\end{cpar}
\end{lcom}
\@x{ AdvanceTime \.{\defeq}}
 \@x{\@s{8.2} \E\, tm \.{\in} Int \.{:} tm \.{\geq} now \.{\land} now \.{'}
 \.{=} tm}
 \@x{\@s{8.2} \.{\land} {\UNCHANGED} {\langle} blockchain ,\, Faulty
 {\rangle}}
\@pvspace{8.0pt}
\begin{lcom}{2.5}
\begin{cpar}{0}{F}{F}{0}{0}{}
 One more process fails. As a result, the \ensuremath{blockchain} may move
 into the faulty zone.
 The light client is not using this action, as the faults are picked in the
 initial state.
 However, this action may be useful when reasoning about fork detection.
\end{cpar}
\end{lcom}
\@x{ OneMoreFault \.{\defeq}}
 \@x{\@s{8.2} \.{\land} \E\, n \.{\in} AllNodes \.{\,\backslash\,} Faulty
 \.{:}}
\@x{\@s{23.41} \.{\land} Faulty \.{'} \.{=} Faulty \.{\cup} \{ n \}}
\@x{\@s{23.41} \.{\land} Faulty \.{'} \.{\neq} AllNodes}
\@y{\@s{0}
 at least process remains non-faulty
}
\@xx{}
\@x{\@s{8.2} \.{\land} {\UNCHANGED} {\langle} now ,\, blockchain {\rangle}}
\@x{}\bottombar\@xx{}
\setboolean{shading}{false}
\begin{lcom}{0}
\begin{cpar}{0}{F}{F}{0}{0}{}
\ensuremath{\.{\,\backslash\,}\.{*}} Modification History
\end{cpar}
\begin{cpar}{0}{F}{F}{0}{0}{}
 \ensuremath{\.{\,\backslash\,}\.{*}} Last modified \ensuremath{Wed}
 \ensuremath{Jun} 10 14:10:54 \ensuremath{CEST} 2020 by \ensuremath{igor
}
\end{cpar}
\begin{cpar}{0}{F}{F}{0}{0}{}
 \ensuremath{\.{\,\backslash\,}\.{*}} Created \ensuremath{Fri}
 \ensuremath{Oct} 11 15:45:11 \ensuremath{CEST} 2019 by \ensuremath{igor
}
\end{cpar}
\end{lcom}

 \clearpage

\tlatex
\setboolean{shading}{true}
\@x{}\moduleLeftDash\@xx{ {\MODULE} Lightclient\_A\_1}\moduleRightDash\@xx{}
\begin{lcom}{2.5}
\begin{cpar}{0}{T}{F}{2.5}{0}{}
\ensuremath{\.{*}
}
\end{cpar}
\begin{cpar}{1}{F}{F}{0}{0}{}
 \ensuremath{\.{*}} A state-machine specification of the lite client,
 following the \ensuremath{English} spec\ensuremath{\.{:}
}
\end{cpar}
\begin{cpar}{0}{F}{F}{0}{0}{}
\ensuremath{\.{*}
}
\end{cpar}
\begin{cpar}{0}{F}{F}{0}{0}{}
\ensuremath{\.{*} https}:\ensuremath{\.{\slsl}github.com}/informalsystems/tendermint-rs/blob/master/docs/spec/lightclient/\ensuremath{verification.md
}
\end{cpar}
\end{lcom}
\@x{ {\EXTENDS} Integers ,\, FiniteSets}
\@pvspace{8.0pt}
\@x{}
\@y{\@s{0}
 the parameters of Light Client
}
\@xx{}
\@x{ {\CONSTANTS}}
\@x{\@s{8.2} TRUSTED\_HEIGHT ,\,}
\@x{\@s{16.4}}
\@y{\@s{0}
 an index of the block header that the light client trusts by social consensus
}
\@xx{}
\@x{\@s{8.2} TARGET\_HEIGHT ,\,}
\@x{\@s{16.4}}
\@y{\@s{0}
 an index of the block header that the light client tries to verify
}
\@xx{}
\@x{\@s{8.2} TRUSTING\_PERIOD ,\,}
\@x{\@s{16.4}}
\@y{\@s{0}
 the period within which the validators are trusted
}
\@xx{}
\@x{\@s{8.2} IS\_PRIMARY\_CORRECT}
\@x{\@s{16.4}}
\@y{\@s{0}
 is primary correct\.{?}
}
\@xx{}
\@pvspace{8.0pt}
\@x{ {\VARIABLES}\@s{24.59}}
\@y{\@s{0}
 see \ensuremath{TypeOK} below for the variable types
}
\@xx{}
\@x{\@s{8.2} state ,\,\@s{40.51}}
\@y{\@s{0}
 the current state of the light client
}
\@xx{}
\@x{\@s{8.2} nextHeight ,\,\@s{14.44}}
\@y{\@s{0}
 the next height to explore by the light client
}
\@xx{}
\@x{\@s{8.2} nprobes\@s{34.67}}
\@y{\@s{0}
 the lite client iteration, or the number of block tests
}
\@xx{}
\@pvspace{8.0pt}
\@x{}
\@y{\@s{0}
 the light store
}
\@xx{}
\@x{ {\VARIABLES}}
\@x{\@s{8.2} fetchedLightBlocks ,\,}
\@y{\@s{0}
 a function from heights to \ensuremath{LightBlocks
}}
\@xx{}
\@x{\@s{8.2} lightBlockStatus ,\,\@s{10.87}}
\@y{\@s{0}
 a function from heights to block statuses
}
\@xx{}
\@x{\@s{8.2} latestVerified\@s{29.28}}
\@y{\@s{0}
 the latest verified block
}
\@xx{}
\@pvspace{8.0pt}
\@x{}
\@y{\@s{0}
 the variables of the lite client
}
\@xx{}
 \@x{ lcvars \.{\defeq} {\langle} state ,\, nextHeight ,\, fetchedLightBlocks
 ,\, lightBlockStatus ,\, latestVerified {\rangle}}
\@pvspace{8.0pt}
\@x{}
\@y{
 ****************** \ensuremath{Blockchain} instance
 **********************************
}
\@xx{}
\@pvspace{8.0pt}
\@x{}
\@y{\@s{0}
 the parameters that are propagated into \ensuremath{Blockchain
}}
\@xx{}
\@x{ {\CONSTANTS}}
\@x{\@s{8.2} AllNodes}
\@x{\@s{16.4}}
\@y{\@s{0}
 a set of all nodes that can act as validators (correct and faulty)
}
\@xx{}
\@pvspace{8.0pt}
\@x{}
\@y{\@s{0}
 the state variables of \ensuremath{Blockchain}, see
 \ensuremath{Blockchain.tla} for the details
}
\@xx{}
\@x{ {\VARIABLES} now ,\, blockchain ,\, Faulty}
\@pvspace{8.0pt}
\@x{}
\@y{\@s{0}
 All the variables of \ensuremath{Blockchain}. For some reason,
 \ensuremath{BC{\bang}vars} does not work
}
\@xx{}
\@x{ bcvars \.{\defeq} {\langle} now ,\, blockchain ,\, Faulty {\rangle}}
\@pvspace{8.0pt}
\begin{lcom}{7.5}
\begin{cpar}{0}{F}{F}{0}{0}{}
Create an instance of \ensuremath{Blockchain}.
 We could write \ensuremath{{\EXTENDS}} \ensuremath{Blockchain}, but then all
 the constants and state variables
 would be hidden inside the \ensuremath{Blockchain} module.
\end{cpar}
\end{lcom}
\@x{ ULTIMATE\_HEIGHT \.{\defeq} TARGET\_HEIGHT \.{+} 1}
\@pvspace{8.0pt}
\@x{ BC \.{\defeq} {\INSTANCE} Blockchain\_A\_1 {\WITH}}
 \@x{\@s{8.2} now \.{\leftarrow} now ,\, blockchain \.{\leftarrow} blockchain
 ,\, Faulty \.{\leftarrow} Faulty}
\@pvspace{8.0pt}
\@x{}
\@y{
 ************************* \ensuremath{Lite} client
 ***********************************
}
\@xx{}
\@pvspace{8.0pt}
\@x{}
\@y{\@s{0}
 the heights on which the light client is working
}
\@xx{}
\@x{ HEIGHTS \.{\defeq} TRUSTED\_HEIGHT \.{\dotdot} TARGET\_HEIGHT}
\@pvspace{8.0pt}
\@x{}
\@y{\@s{0}
 the control states of the lite client
}
\@xx{}
 \@x{ States \.{\defeq} \{\@w{working} ,\,\@w{finishedSuccess}
 ,\,\@w{finishedFailure} \}}
\@pvspace{8.0pt}
\begin{lcom}{2.5}
\begin{cpar}{0}{T}{F}{2.5}{0}{}
*
\end{cpar}
\begin{cpar}{1}{F}{F}{0}{0}{}
Check the precondition of \ensuremath{ValidAndVerified}.
\end{cpar}
\vshade{5.0}
\begin{cpar}{0}{F}{F}{0}{0}{}
[LCV-FUNC-VALID.\ensuremath{1{\coloncolon}}TLA-PRE.1]
\end{cpar}
\end{lcom}
\@x{ ValidAndVerifiedPre ( trusted ,\, untrusted ) \.{\defeq}}
\@x{\@s{8.2} \.{\LET} thdr\@s{2.04} \.{\defeq} trusted . header}
\@x{\@s{28.59} uhdr \.{\defeq} untrusted . header}
\@x{\@s{8.2} \.{\IN}}
\@x{\@s{8.2} \.{\land} BC {\bang} InTrustingPeriod ( thdr )}
\@x{\@s{8.2} \.{\land} thdr . height \.{<} uhdr . height}
\@x{\@s{19.31}}
\@y{\@s{0}
 the trusted block has been created earlier (no drift here)
}
\@xx{}
\@x{\@s{8.2} \.{\land} thdr . time \.{\leq} uhdr . time}
\@x{\@s{8.2} \.{\land} untrusted . Commits \.{\subseteq} uhdr . VS}
\@x{\@s{8.2} \.{\land} \.{\LET} TP \.{\defeq} Cardinality ( uhdr . VS )}
\@x{\@s{39.71} SP\@s{1.53} \.{\defeq} Cardinality ( untrusted . Commits )}
\@x{\@s{19.31} \.{\IN}}
\@x{\@s{19.31} 3 \.{*} SP \.{>} 2 \.{*} TP}
 \@x{\@s{8.2} \.{\land} thdr . height \.{+} 1 \.{=} uhdr . height \.{\implies}
 thdr . NextVS \.{=} uhdr . VS}
\begin{lcom}{0}
\begin{cpar}{0}{T}{F}{12.5}{0}{}
 As we do not have explicit hashes we ignore these three checks of the
 \ensuremath{English} spec:
\end{cpar}
\vshade{5.0}
\begin{cpar}{0}{T}{T}{0}{0}{
1}
 . ``\ensuremath{trusted.Commit} is a commit is for the header
 \ensuremath{trusted.Header},
 \ensuremath{i.e}. it contains the correct hash of the header''.
\end{cpar}
\begin{cpar}{1}{F}{F}{0}{0}{}
2. \ensuremath{untrusted.Validators \.{=} hash(untrusted.Header.Validators)
}
\end{cpar}
\begin{cpar}{0}{F}{F}{0}{0}{}
 3. \ensuremath{untrusted.NextValidators \.{=}
 hash(untrusted.Header.NextValidators)
}
\end{cpar}
\vshade{10.0}
\begin{cpar}{1}{F}{F}{0}{0}{}
\ensuremath{\.{*}
}
\end{cpar}
\begin{cpar}{0}{T}{F}{5.0}{0}{}
 \ensuremath{\.{*}} Check that the commits in an \ensuremath{untrusted} block
 form 1/3 of the next \ensuremath{validators
}
\end{cpar}
\begin{cpar}{0}{F}{F}{0}{0}{}
\ensuremath{\.{*}} in a trusted header.
\end{cpar}
\end{lcom}
\@x{ SignedByOneThirdOfTrusted ( trusted ,\, untrusted ) \.{\defeq}}
 \@x{\@s{8.2} \.{\LET} TP \.{\defeq} Cardinality ( trusted . header . NextVS
 )}
 \@x{\@s{28.59} SP\@s{1.53} \.{\defeq} Cardinality ( untrusted . Commits
 \.{\cap} trusted . header . NextVS )}
\@x{\@s{8.2} \.{\IN}}
\@x{\@s{8.2} 3 \.{*} SP \.{>} TP}
\@pvspace{8.0pt}
\begin{lcom}{2.5}
\begin{cpar}{0}{T}{F}{2.5}{0}{}
*
\end{cpar}
\begin{cpar}{1}{F}{F}{0}{0}{}
 Check, whether an \ensuremath{untrusted} block is valid and verifiable
 \ensuremath{w.r.t}. a trusted header.
\end{cpar}
\vshade{5.0}
\begin{cpar}{0}{F}{F}{0}{0}{}
[LCV-FUNC-VALID.\ensuremath{1{\coloncolon}}TLA.1]
\end{cpar}
\end{lcom}
\@x{ ValidAndVerified ( trusted ,\, untrusted ) \.{\defeq}}
\@x{\@s{16.4} {\IF} {\lnot} ValidAndVerifiedPre ( trusted ,\, untrusted )}
\@x{\@s{16.4} \.{\THEN}\@w{FAILED\_VERIFICATION}}
 \@x{\@s{16.4} \.{\ELSE} {\IF} {\lnot} BC {\bang} InTrustingPeriod ( untrusted
 . header )}
\begin{lcom}{17.5}
\begin{cpar}{0}{F}{F}{0}{0}{}
We leave the following test for the documentation purposes.
 The implementation should do this test, as signature verification may be
 slow.
 In the TLA+ specification, \ensuremath{ValidAndVerified} happens in no time.
\end{cpar}
\end{lcom}
\@x{\@s{16.4} \.{\THEN}\@w{FAILED\_TRUSTING\_PERIOD}}
 \@x{\@s{16.4} \.{\ELSE} {\IF} untrusted . header . height \.{=} trusted .
 header . height \.{+} 1}
\@x{\@s{63.96} \.{\lor} SignedByOneThirdOfTrusted ( trusted ,\, untrusted )}
\@x{\@s{47.71} \.{\THEN}\@w{OK}}
\@x{\@s{47.71} \.{\ELSE}\@w{CANNOT\_VERIFY}}
\@pvspace{8.0pt}
\begin{lcom}{2.5}
\begin{cpar}{0}{F}{F}{0}{0}{}
Initial states of the light client.
 Initially, only the trusted light block is present.
\end{cpar}
\end{lcom}
\@x{ LCInit \.{\defeq}}
\@x{\@s{16.4} \.{\land}\@s{2.73} state \.{=}\@w{working}}
\@x{\@s{16.4} \.{\land}\@s{2.73} nextHeight \.{=} TARGET\_HEIGHT}
\@x{\@s{16.4} \.{\land}\@s{2.73} nprobes \.{=} 0\@s{4.1}}
\@y{\@s{0}
 no tests have been done so far
}
\@xx{}
 \@x{\@s{16.4} \.{\land}\@s{2.73} \.{\LET} trustedBlock \.{\defeq} blockchain
 [ TRUSTED\_HEIGHT ]}
 \@x{\@s{50.64} trustedLightBlock \.{\defeq} [ header \.{\mapsto} trustedBlock
 ,\, Commits \.{\mapsto} AllNodes ]}
\@x{\@s{30.24} \.{\IN}}
\@x{\@s{34.34}}
\@y{\@s{0}
 initially, \ensuremath{fetchedLightBlocks} is a function of one element,
 \ensuremath{i.e}., \ensuremath{TRUSTED\_HEIGHT
}}
\@xx{}
 \@x{\@s{34.34} \.{\land} fetchedLightBlocks \.{=} [ h \.{\in} \{
 TRUSTED\_HEIGHT \} \.{\mapsto} trustedLightBlock ]}
\@x{\@s{34.34}}
\@y{\@s{0}
 initially, \ensuremath{lightBlockStatus} is a function of one element,
 \ensuremath{i.e}., \ensuremath{TRUSTED\_HEIGHT
}}
\@xx{}
 \@x{\@s{34.34} \.{\land} lightBlockStatus \.{=} [ h \.{\in} \{
 TRUSTED\_HEIGHT \} \.{\mapsto}\@w{StateVerified} ]}
\@x{\@s{34.34}}
\@y{\@s{0}
 the latest verified block the the trusted block
}
\@xx{}
\@x{\@s{34.34} \.{\land} latestVerified \.{=} trustedLightBlock}
\@pvspace{8.0pt}
\@x{}
\@y{\@s{0}
 block should contain a copy of the block from the reference chain, with a
 matching commit
}
\@xx{}
\@x{ CopyLightBlockFromChain ( block ,\, height ) \.{\defeq}}
\@x{\@s{16.4} \.{\LET} ref \.{\defeq} blockchain [ height ]}
\@x{\@s{36.79} lastCommit \.{\defeq}}
\@x{\@s{44.99} {\IF} height \.{<} ULTIMATE\_HEIGHT}
\@x{\@s{44.99} \.{\THEN} blockchain [ height \.{+} 1 ] . lastCommit}
\@x{\@s{53.19}}
\@y{\@s{0}
 for the ultimate block, which we never use, as \ensuremath{ULTIMATE\_HEIGHT
 \.{=} TARGET\_HEIGHT \.{+} 1
}}
\@xx{}
\@x{\@s{44.99} \.{\ELSE} blockchain [ height ] . VS}
\@x{\@s{16.4} \.{\IN}}
 \@x{\@s{16.4} block \.{=} [ header \.{\mapsto} ref ,\, Commits \.{\mapsto}
 lastCommit ]}
\@pvspace{8.0pt}
\@x{}
\@y{\@s{0}
 Either the primary is correct and the block comes from the reference chain,
}
\@xx{}
\@x{}
\@y{\@s{0}
 or the block is produced by a faulty primary.
}
\@xx{}
\@x{}
\@y{}
\@xx{}
\@x{}
\@y{\@s{0}
 [LCV-FUNC-FETCH.\ensuremath{1{\coloncolon}}TLA.1]
}
\@xx{}
\@x{ FetchLightBlockInto ( block ,\, height ) \.{\defeq}}
\@x{\@s{16.4} {\IF} IS\_PRIMARY\_CORRECT}
\@x{\@s{16.4} \.{\THEN} CopyLightBlockFromChain ( block ,\, height )}
 \@x{\@s{16.4} \.{\ELSE} BC {\bang} IsLightBlockAllowedByDigitalSignatures (
 height ,\, block )}
\@pvspace{8.0pt}
\@x{}
\@y{\@s{0}
 add a block into the light store
}
\@xx{}
\@x{}
\@y{}
\@xx{}
\@x{}
\@y{\@s{0}
 [LCV-FUNC-UPDATE.\ensuremath{1{\coloncolon}}TLA.1]
}
\@xx{}
\@x{ LightStoreUpdateBlocks ( lightBlocks ,\, block ) \.{\defeq}}
\@x{\@s{16.4} \.{\LET} ht \.{\defeq} block . header . height \.{\IN}}
 \@x{\@s{16.4} [ h \.{\in} {\DOMAIN} lightBlocks \.{\cup} \{ ht \}
 \.{\mapsto}}
 \@x{\@s{29.15} {\IF} h \.{=} ht \.{\THEN} block \.{\ELSE} lightBlocks [ h ]
 ]}
\@pvspace{8.0pt}
\@x{}
\@y{\@s{0}
 update the state of a light block
}
\@xx{}
\@x{}
\@y{}
\@xx{}
\@x{}
\@y{\@s{0}
 [LCV-FUNC-UPDATE.\ensuremath{1{\coloncolon}}TLA.1]
}
\@xx{}
\@x{ LightStoreUpdateStates ( statuses ,\, ht ,\, blockState ) \.{\defeq}}
\@x{\@s{16.4} [ h \.{\in} {\DOMAIN} statuses \.{\cup} \{ ht \} \.{\mapsto}}
 \@x{\@s{29.15} {\IF} h \.{=} ht \.{\THEN} blockState \.{\ELSE} statuses [ h ]
 ]}
\@pvspace{8.0pt}
\@x{}
\@y{\@s{0}
 Check, whether \ensuremath{newHeight} is a possible next height for the
 light client.
}
\@xx{}
\@x{}
\@y{}
\@xx{}
\@x{}
\@y{\@s{0}
 [LCV-FUNC-SCHEDULE.\ensuremath{1{\coloncolon}}TLA.1]
}
\@xx{}
 \@x{ CanScheduleTo ( newHeight ,\, pLatestVerified ,\, pNextHeight ,\,
 pTargetHeight ) \.{\defeq}}
 \@x{\@s{16.4} \.{\LET} ht \.{\defeq} pLatestVerified . header . height
 \.{\IN}}
\@x{\@s{16.4} \.{\lor} \.{\land} ht \.{=} pNextHeight}
\@x{\@s{27.51} \.{\land} ht \.{<} pTargetHeight}
\@x{\@s{27.51} \.{\land} pNextHeight \.{<} newHeight}
\@x{\@s{27.51} \.{\land} newHeight \.{\leq} pTargetHeight}
\@x{\@s{16.4} \.{\lor} \.{\land} ht \.{<} pNextHeight}
\@x{\@s{27.51} \.{\land} ht \.{<} pTargetHeight}
\@x{\@s{27.51} \.{\land} ht \.{<} newHeight}
\@x{\@s{27.51} \.{\land} newHeight \.{<} pNextHeight}
\@x{\@s{16.4} \.{\lor} \.{\land} ht \.{=} pTargetHeight}
\@x{\@s{27.51} \.{\land} newHeight \.{=} pTargetHeight}
\@pvspace{8.0pt}
\@x{}
\@y{\@s{0}
 The loop of \ensuremath{VerifyToTarget}.
}
\@xx{}
\@x{}
\@y{}
\@xx{}
\@x{}
\@y{\@s{0}
 [LCV-FUNC-MAIN.\ensuremath{1{\coloncolon}}TLA-LOOP.1]
}
\@xx{}
\@x{ VerifyToTargetLoop \.{\defeq}}
\@x{\@s{24.59}}
\@y{\@s{0}
 the loop condition is true
}
\@xx{}
 \@x{\@s{16.4} \.{\land} latestVerified . header . height \.{<}
 TARGET\_HEIGHT}
\@x{\@s{24.59}}
\@y{\@s{0}
 pick a light block, which will be constrained later
}
\@xx{}
\@x{\@s{16.4} \.{\land} \E\, current \.{\in} BC {\bang} LightBlocks \.{:}}
\@x{\@s{31.61}}
\@y{\@s{0}
 Get next \ensuremath{LightBlock} for verification
}
\@xx{}
 \@x{\@s{31.61} \.{\land} {\IF} nextHeight \.{\in} {\DOMAIN}
 fetchedLightBlocks}
\@x{\@s{42.72} \.{\THEN}}
\@y{\@s{0}
 copy the block from the light store
}
\@xx{}
\@x{\@s{74.03} \.{\land} current \.{=} fetchedLightBlocks [ nextHeight ]}
\@x{\@s{74.03} \.{\land} {\UNCHANGED} fetchedLightBlocks}
\@x{\@s{42.72} \.{\ELSE}}
\@y{\@s{0}
 retrieve a light block and save it in the light store
}
\@xx{}
\@x{\@s{74.03} \.{\land} FetchLightBlockInto ( current ,\, nextHeight )}
 \@x{\@s{74.03} \.{\land} fetchedLightBlocks \.{'} \.{=}
 LightStoreUpdateBlocks ( fetchedLightBlocks ,\, current )}
\@x{\@s{31.61}}
\@y{\@s{0}
 Record that one more probe has been done (for complexity and model checking)
}
\@xx{}
\@x{\@s{31.61} \.{\land} nprobes \.{'} \.{=} nprobes \.{+} 1}
\@x{\@s{31.61}}
\@y{\@s{0}
 Verify the current block
}
\@xx{}
 \@x{\@s{31.61} \.{\land} \.{\LET} verdict \.{\defeq} ValidAndVerified (
 latestVerified ,\, current ) \.{\IN}}
\@x{\@s{42.72}}
\@y{\@s{0}
 Decide whether/how to continue
}
\@xx{}
\@x{\@s{42.72} {\CASE} verdict \.{=}\@w{OK} \.{\rightarrow}}
 \@x{\@s{55.02} \.{\land} lightBlockStatus \.{'} \.{=} LightStoreUpdateStates
 ( lightBlockStatus ,\, nextHeight ,\,\@w{StateVerified} )}
\@x{\@s{55.02} \.{\land} latestVerified \.{'} \.{=} current}
\@x{\@s{55.02} \.{\land} state \.{'} \.{=}}
 \@x{\@s{78.43} {\IF} latestVerified \.{'} . header . height \.{<}
 TARGET\_HEIGHT}
\@x{\@s{78.43} \.{\THEN}\@w{working}}
\@x{\@s{78.43} \.{\ELSE}\@w{finishedSuccess}}
\@x{\@s{55.02} \.{\land} \E\, newHeight \.{\in} HEIGHTS \.{:}}
 \@x{\@s{66.13} \.{\land} CanScheduleTo ( newHeight ,\, current ,\, nextHeight
 ,\, TARGET\_HEIGHT )}
\@x{\@s{66.13} \.{\land} nextHeight \.{'} \.{=} newHeight}
\@pvspace{8.0pt}
 \@x{\@s{42.72} {\Box}\@s{4.82} verdict \.{=}\@w{CANNOT\_VERIFY}
 \.{\rightarrow}}
\begin{lcom}{40.0}
\begin{cpar}{0}{F}{F}{0}{0}{}
do nothing: the light block current passed validation, but the validator
 set is too different to verify it. We keep the state of
 current at \ensuremath{StateUnverified}. For a later iteration, Schedule
 might decide to try verification of that light block again.
\end{cpar}
\end{lcom}
 \@x{\@s{57.4} \.{\land} lightBlockStatus \.{'} \.{=} LightStoreUpdateStates (
 lightBlockStatus ,\, nextHeight ,\,\@w{StateUnverified} )}
\@x{\@s{57.4} \.{\land} \E\, newHeight \.{\in} HEIGHTS \.{:}}
 \@x{\@s{68.51} \.{\land} CanScheduleTo ( newHeight ,\, latestVerified ,\,
 nextHeight ,\, TARGET\_HEIGHT )}
\@x{\@s{68.51} \.{\land} nextHeight \.{'} \.{=} newHeight}
 \@x{\@s{57.4} \.{\land} {\UNCHANGED} {\langle} latestVerified ,\, state
 {\rangle}}
\@pvspace{8.0pt}
\@x{\@s{45.1} {\Box}\@s{4.82} {\OTHER} \.{\rightarrow}}
\@x{\@s{57.4}}
\@y{\@s{0}
 verdict is some error code
}
\@xx{}
 \@x{\@s{57.4} \.{\land} lightBlockStatus \.{'} \.{=} LightStoreUpdateStates (
 lightBlockStatus ,\, nextHeight ,\,\@w{StateFailed} )}
\@x{\@s{57.4} \.{\land} state \.{'} \.{=}\@w{finishedFailure}}
 \@x{\@s{57.4} \.{\land} {\UNCHANGED} {\langle} latestVerified ,\, nextHeight
 {\rangle}}
\@pvspace{8.0pt}
\@x{}
\@y{\@s{0}
 The terminating condition of \ensuremath{VerifyToTarget}.
}
\@xx{}
\@x{}
\@y{}
\@xx{}
\@x{}
\@y{\@s{0}
 [LCV-FUNC-MAIN.\ensuremath{1{\coloncolon}}TLA-LOOPCOND.1]
}
\@xx{}
\@x{ VerifyToTargetDone \.{\defeq}}
 \@x{\@s{16.4} \.{\land} latestVerified . header . height \.{\geq}
 TARGET\_HEIGHT}
\@x{\@s{16.4} \.{\land} state \.{'} \.{=}\@w{finishedSuccess}}
 \@x{\@s{16.4} \.{\land} {\UNCHANGED} {\langle} nextHeight ,\, nprobes ,\,
 fetchedLightBlocks ,\, lightBlockStatus ,\, latestVerified {\rangle}}
\@pvspace{8.0pt}
\@x{}
\@y{
 ******************** \ensuremath{Lite} client \ensuremath{\.{+} Blockchain}
 ******************
}
\@xx{}
\@x{ Init \.{\defeq}}
\@x{\@s{16.4}}
\@y{\@s{0}
 the \ensuremath{blockchain} is initialized immediately to the
 \ensuremath{ULTIMATE\_HEIGHT
}}
\@xx{}
\@x{\@s{16.4} \.{\land} BC {\bang} InitToHeight}
\@x{\@s{16.4}}
\@y{\@s{0}
 the light client starts
}
\@xx{}
\@x{\@s{16.4} \.{\land} LCInit}
\@pvspace{8.0pt}
\begin{lcom}{5.0}
\begin{cpar}{0}{F}{F}{0}{0}{}
The system step is very simple.
 The light client is either executing \ensuremath{VerifyToTarget}, or it has
 terminated.
 (In the latter case, a model checker reports a deadlock.)
 Simultaneously, the global clock may advance.
\end{cpar}
\end{lcom}
\@x{ Next \.{\defeq}}
\@x{\@s{16.4} \.{\land} state \.{=}\@w{working}}
\@x{\@s{16.4} \.{\land} VerifyToTargetLoop \.{\lor} VerifyToTargetDone}
\@x{\@s{16.4} \.{\land} BC {\bang} AdvanceTime}
\@y{\@s{0}
 the global clock is advanced by zero or more time units
}
\@xx{}
\@pvspace{8.0pt}
\@x{}
\@y{
 ************************ Types *****************************************
}
\@xx{}
\@x{ TypeOK \.{\defeq}}
\@x{\@s{16.4} \.{\land}\@s{9.74} state \.{\in} States}
\@x{\@s{16.4} \.{\land}\@s{9.74} nextHeight \.{\in} HEIGHTS}
 \@x{\@s{16.4} \.{\land}\@s{9.74} latestVerified \.{\in} BC {\bang}
 LightBlocks}
\@x{\@s{16.4} \.{\land}\@s{9.74} \E\, HS \.{\in} {\SUBSET} HEIGHTS \.{:}}
 \@x{\@s{41.35} \.{\land} fetchedLightBlocks \.{\in} [ HS \.{\rightarrow} BC
 {\bang} LightBlocks ]}
\@x{\@s{41.35} \.{\land} lightBlockStatus}
 \@x{\@s{60.66} \.{\in} [ HS \.{\rightarrow} \{\@w{StateVerified}
 ,\,\@w{StateUnverified} ,\,\@w{StateFailed} \} ]}
\@pvspace{8.0pt}
\@x{}
\@y{
 ************************ Properties *****************************************
}
\@xx{}
\@pvspace{8.0pt}
\@x{}
\@y{\@s{0}
 The properties to check
}
\@xx{}
\@x{}
\@y{\@s{0}
 this invariant candidate is false
}
\@xx{}
\@x{ NeverFinish \.{\defeq}}
\@x{\@s{16.4} state \.{=}\@w{working}}
\@pvspace{8.0pt}
\@x{}
\@y{\@s{0}
 this invariant candidate is false
}
\@xx{}
\@x{ NeverFinishNegative \.{\defeq}}
\@x{\@s{16.4} state \.{\neq}\@w{finishedFailure}}
\@pvspace{8.0pt}
\@x{}
\@y{\@s{0}
 This invariant holds true, when the primary is correct.
}
\@xx{}
\@x{}
\@y{\@s{0}
 This invariant candidate is false when the primary is faulty.
}
\@xx{}
\@x{ NeverFinishNegativeWhenTrusted \.{\defeq}}
\@x{\@s{16.4}}
\@y{
 (\ensuremath{minTrustedHeight \.{\leq} TRUSTED\_HEIGHT})
}
\@xx{}
\@x{\@s{16.4} BC {\bang} InTrustingPeriod ( blockchain [ TRUSTED\_HEIGHT ] )}
\@x{\@s{32.04} \.{\implies} state \.{\neq}\@w{finishedFailure}}
\@pvspace{8.0pt}
\@x{}
\@y{\@s{0}
 this invariant candidate is false
}
\@xx{}
\@x{ NeverFinishPositive \.{\defeq}}
\@x{\@s{16.4} state \.{\neq}\@w{finishedSuccess}}
\@pvspace{8.0pt}
\begin{lcom}{5.0}
\begin{cpar}{0}{F}{F}{0}{0}{}
*
 Correctness states that all the obtained headers are exactly like in the
 \ensuremath{blockchain}.
\end{cpar}
\vshade{5.0}
\begin{cpar}{0}{F}{F}{0}{0}{}
 It is always the case that every verified header in \ensuremath{LightStore}
 was generated by
 an instance of \ensuremath{Tendermint} consensus.
\end{cpar}
\vshade{5.0}
\begin{cpar}{0}{F}{F}{0}{0}{}
[LCV-DIST-SAFE.\ensuremath{1{\coloncolon}}CORRECTNESS-INV.1]
\end{cpar}
\end{lcom}
\@x{ CorrectnessInv \.{\defeq}}
\@x{\@s{16.4} \A\, h \.{\in} {\DOMAIN} fetchedLightBlocks \.{:}}
\@x{\@s{27.72} lightBlockStatus [ h ] \.{=}\@w{StateVerified} \.{\implies}}
\@x{\@s{44.12} fetchedLightBlocks [ h ] . header \.{=} blockchain [ h ]}
\@pvspace{8.0pt}
\begin{lcom}{2.5}
\begin{cpar}{0}{T}{F}{2.5}{0}{}
*
\end{cpar}
\begin{cpar}{1}{F}{F}{0}{0}{}
 Check that the sequence of the headers in \ensuremath{storedLightBlocks}
 satisfies \ensuremath{ValidAndVerified \.{=} \@w{OK}} pairwise
 This property is easily violated, whenever a header cannot be trusted
 anymore.
\end{cpar}
\end{lcom}
\@x{ StoredHeadersAreVerifiedInv \.{\defeq}}
\@x{\@s{16.4} state \.{=}\@w{finishedSuccess}}
\@x{\@s{32.8} \.{\implies}}
\@x{\@s{32.8} \A\, lh ,\, rh \.{\in} {\DOMAIN} fetchedLightBlocks \.{:}}
\@y{\@s{0}
 for every pair of different stored headers
}
\@xx{}
\@x{\@s{44.12} \.{\lor} lh\@s{0.99} \.{\geq} rh}
\@x{\@s{55.23}}
\@y{\@s{0}
 either there is a header between them
}
\@xx{}
\@x{\@s{44.12} \.{\lor} \E\, mh \.{\in} {\DOMAIN} fetchedLightBlocks \.{:}}
\@x{\@s{59.33} lh \.{<} mh \.{\land} mh \.{<} rh}
\@x{\@s{55.23}}
\@y{\@s{0}
 or we can verify the right one using the left one
}
\@xx{}
 \@x{\@s{44.12} \.{\lor}\@w{OK} \.{=} ValidAndVerified ( fetchedLightBlocks [
 lh ] ,\, fetchedLightBlocks [ rh ] )}
\@pvspace{8.0pt}
\@x{}
\@y{\@s{0}
 An improved version of \ensuremath{StoredHeadersAreSound}, assuming that a
 header may be not trusted.
}
\@xx{}
\@x{}
\@y{\@s{0}
 This invariant candidate is also violated,
}
\@xx{}
\@x{}
\@y{\@s{0}
 as there may be some unverified blocks left in the middle.
}
\@xx{}
\@x{ StoredHeadersAreVerifiedOrNotTrustedInv \.{\defeq}}
\@x{\@s{16.4} state \.{=}\@w{finishedSuccess}}
\@x{\@s{32.8} \.{\implies}}
\@x{\@s{32.8} \A\, lh ,\, rh \.{\in} {\DOMAIN} fetchedLightBlocks \.{:}}
\@y{\@s{0}
 for every pair of different stored headers
}
\@xx{}
\@x{\@s{44.12} \.{\lor} lh\@s{0.99} \.{\geq} rh}
\@x{\@s{55.23}}
\@y{\@s{0}
 either there is a header between them
}
\@xx{}
\@x{\@s{44.12} \.{\lor} \E\, mh \.{\in} {\DOMAIN} fetchedLightBlocks \.{:}}
\@x{\@s{59.33} lh \.{<} mh \.{\land} mh \.{<} rh}
\@x{\@s{55.23}}
\@y{\@s{0}
 or we can verify the right one using the left one
}
\@xx{}
 \@x{\@s{44.12} \.{\lor}\@w{OK} \.{=} ValidAndVerified ( fetchedLightBlocks [
 lh ] ,\, fetchedLightBlocks [ rh ] )}
\@x{\@s{55.23}}
\@y{\@s{0}
 or the left header is outside the trusting period, so no guarantees
}
\@xx{}
 \@x{\@s{44.12} \.{\lor} {\lnot} BC {\bang} InTrustingPeriod (
 fetchedLightBlocks [ lh ] . header )}
\@pvspace{8.0pt}
\begin{lcom}{2.5}
\begin{cpar}{0}{T}{F}{2.5}{0}{}
\ensuremath{\.{*}
}
\end{cpar}
\begin{cpar}{1}{F}{F}{0}{0}{}
 \ensuremath{\.{*}} An improved version of
 \ensuremath{StoredHeadersAreSoundOrNotTrusted,\,
}
\end{cpar}
\begin{cpar}{0}{F}{F}{0}{0}{}
 \ensuremath{\.{*}} checking the property only for the verified
 headers\ensuremath{.
}
\end{cpar}
\begin{cpar}{0}{F}{F}{0}{0}{}
\ensuremath{\.{*}} This invariant holds true.
\end{cpar}
\end{lcom}
\@x{ ProofOfChainOfTrustInv \.{\defeq}}
\@x{\@s{16.4} state \.{=}\@w{finishedSuccess}}
\@x{\@s{32.8} \.{\implies}}
\@x{\@s{32.8} \A\, lh ,\, rh \.{\in} {\DOMAIN} fetchedLightBlocks \.{:}}
\@x{\@s{58.66}}
\@y{\@s{0}
 for every pair of stored headers that have been verified
}
\@xx{}
\@x{\@s{44.12} \.{\lor} lh \.{\geq} rh}
 \@x{\@s{44.12} \.{\lor} lightBlockStatus [ lh ]\@s{1.66}
 \.{=}\@w{StateUnverified}}
\@x{\@s{44.12} \.{\lor} lightBlockStatus [ rh ] \.{=}\@w{StateUnverified}}
\@x{\@s{55.23}}
\@y{\@s{0}
 either there is a header between them
}
\@xx{}
\@x{\@s{44.12} \.{\lor} \E\, mh \.{\in} {\DOMAIN} fetchedLightBlocks \.{:}}
 \@x{\@s{59.33} lh \.{<} mh \.{\land} mh \.{<} rh \.{\land} lightBlockStatus [
 mh ] \.{=}\@w{StateVerified}}
\@x{\@s{55.23}}
\@y{\@s{0}
 or the left header is outside the trusting period, so no guarantees
}
\@xx{}
 \@x{\@s{44.12} \.{\lor} {\lnot} ( BC {\bang} InTrustingPeriod (
 fetchedLightBlocks [ lh ] . header ) )}
\@x{\@s{55.23}}
\@y{\@s{0}
 or we can verify the right one using the left one
}
\@xx{}
 \@x{\@s{44.12} \.{\lor}\@w{OK} \.{=} ValidAndVerified ( fetchedLightBlocks [
 lh ] ,\, fetchedLightBlocks [ rh ] )}
\@pvspace{8.0pt}
\begin{lcom}{2.5}
\begin{cpar}{0}{T}{F}{2.5}{0}{}
\ensuremath{\.{*}
}
\end{cpar}
\begin{cpar}{1}{F}{F}{0}{0}{}
 \ensuremath{\.{*}} When the light client terminates, there are no failed
 blocks. (Otherwise, someone lied to us.)
\end{cpar}
\end{lcom}
\@x{ NoFailedBlocksOnSuccessInv \.{\defeq}}
\@x{\@s{16.4} state \.{=}\@w{finishedSuccess} \.{\implies}}
\@x{\@s{32.8} \A\, h \.{\in} {\DOMAIN} fetchedLightBlocks \.{:}}
\@x{\@s{44.12} lightBlockStatus [ h ] \.{\neq}\@w{StateFailed}}
\@pvspace{8.0pt}
\@x{}
\@y{\@s{0}
 This property states that whenever the light client finishes with a positive
 outcome,
}
\@xx{}
\@x{}
\@y{\@s{0}
 the trusted header is still within the trusting period.
}
\@xx{}
\@x{}
\@y{\@s{0}
 We expect this property to be violated. And \ensuremath{Apalache} shows us a
 counterexample.
}
\@xx{}
\@x{ PositiveBeforeTrustedHeaderExpires \.{\defeq}}
 \@x{\@s{16.4} ( state \.{=}\@w{finishedSuccess} ) \.{\implies} BC {\bang}
 InTrustingPeriod ( blockchain [ TRUSTED\_HEIGHT ] )}
\@pvspace{8.0pt}
\@x{}
\@y{\@s{0}
 If the primary is correct and the initial trusted block has not expired,
}
\@xx{}
\@x{}
\@y{\@s{0}
 then whenever the algorithm terminates, it reports ``success''
}
\@xx{}
\@x{ CorrectPrimaryAndTimeliness \.{\defeq}}
 \@x{\@s{8.2} ( BC {\bang} InTrustingPeriod ( blockchain [ TRUSTED\_HEIGHT ]
 )}
 \@x{\@s{16.18} \.{\land} state \.{\neq}\@w{working} \.{\land}
 IS\_PRIMARY\_CORRECT ) \.{\implies}}
\@x{\@s{24.38} state \.{=}\@w{finishedSuccess}}
\@pvspace{8.0pt}
\begin{lcom}{5.0}
\begin{cpar}{0}{F}{F}{0}{0}{}
*
 If the primary is correct and there is a trusted block that has not expired,
 then whenever the algorithm terminates, it reports ``success''.
\end{cpar}
\vshade{5.0}
\begin{cpar}{0}{F}{F}{0}{0}{}
[LCV-DIST-LIVE.\ensuremath{1{\coloncolon}}SUCCESS-CORR-PRIMARY-CHAIN-OF-TRUST.1]
\end{cpar}
\end{lcom}
\@x{ SuccessOnCorrectPrimaryAndChainOfTrust \.{\defeq}}
\@x{\@s{8.2} ( \E\, h \.{\in} {\DOMAIN} fetchedLightBlocks \.{:}}
 \@x{\@s{25.18} lightBlockStatus [ h ] \.{=}\@w{StateVerified} \.{\land} BC
 {\bang} InTrustingPeriod ( blockchain [ h ] )}
 \@x{\@s{12.08} \.{\land} state \.{\neq}\@w{working} \.{\land}
 IS\_PRIMARY\_CORRECT ) \.{\implies}}
\@x{\@s{20.28} state \.{=}\@w{finishedSuccess}}
\@pvspace{8.0pt}
\@x{}
\@y{\@s{0}
 Lite Client Completeness: If header \ensuremath{h} was correctly generated
 by an instance
}
\@xx{}
\@x{}
\@y{\@s{0}
 of \ensuremath{Tendermint} consensus (and its age is less than the trusting
 period),
}
\@xx{}
\@x{}
\@y{\@s{0}
 then the lite client should eventually set \ensuremath{trust(h)} to true.
}
\@xx{}
\@x{}
\@y{}
\@xx{}
\@x{}
\@y{\@s{0}
 Note that Completeness assumes that the lite client communicates with a
 correct full node.
}
\@xx{}
\@x{}
\@y{}
\@xx{}
\@x{}
\@y{\@s{0}
 We decompose completeness into Termination (liveness) and Precision (safety).
}
\@xx{}
\@x{}
\@y{\@s{0}
 Once again, Precision is an inverse version of the safety property in
 Completeness,
}
\@xx{}
\@x{}
\@y{\@s{0}
 as A \ensuremath{\.{\implies} B} is logically equivalent to
 \ensuremath{{\lnot}B \.{\implies}} \ensuremath{\sim}A.
}
\@xx{}
\@x{ PrecisionInv \.{\defeq}}
\@x{\@s{16.4} ( state \.{=}\@w{finishedFailure} )}
 \@x{\@s{24.38} \.{\implies} \.{\lor} {\lnot} BC {\bang} InTrustingPeriod (
 blockchain [ TRUSTED\_HEIGHT ] )}
\@y{\@s{0}
 outside of the trusting period
}
\@xx{}
\@x{\@s{39.94} \.{\lor} \E\, h \.{\in} {\DOMAIN} fetchedLightBlocks \.{:}}
 \@x{\@s{51.05} \.{\LET} lightBlock \.{\defeq} fetchedLightBlocks [ h ]
 \.{\IN}}
\@x{\@s{75.55}}
\@y{\@s{0}
 the full node lied to the lite client about the block header
}
\@xx{}
 \@x{\@s{59.25} \.{\lor}\@s{5.18} lightBlock . header \.{\neq} blockchain [ h
 ]}
\@x{\@s{75.55}}
\@y{\@s{0}
 the full node lied to the lite client about the commits
}
\@xx{}
 \@x{\@s{59.25} \.{\lor}\@s{5.18} lightBlock . Commits \.{\neq} lightBlock .
 header . VS}
\@pvspace{8.0pt}
\@x{}
\@y{\@s{0}
 the old invariant that was found to be buggy by \ensuremath{TLC
}}
\@xx{}
\@x{ PrecisionBuggyInv \.{\defeq}}
\@x{\@s{16.4} ( state \.{=}\@w{finishedFailure} )}
 \@x{\@s{24.38} \.{\implies} \.{\lor} {\lnot} BC {\bang} InTrustingPeriod (
 blockchain [ TRUSTED\_HEIGHT ] )}
\@y{\@s{0}
 outside of the trusting period
}
\@xx{}
\@x{\@s{39.94} \.{\lor} \E\, h \.{\in} {\DOMAIN} fetchedLightBlocks \.{:}}
 \@x{\@s{51.05} \.{\LET} lightBlock \.{\defeq} fetchedLightBlocks [ h ]
 \.{\IN}}
\@x{\@s{51.05}}
\@y{\@s{0}
 the full node lied to the lite client about the block header
}
\@xx{}
\@x{\@s{51.05} lightBlock . header \.{\neq} blockchain [ h ]}
\@pvspace{8.0pt}
\@x{}
\@y{\@s{0}
 the worst complexity
}
\@xx{}
\@x{ Complexity \.{\defeq}}
 \@x{\@s{16.4} \.{\LET} N \.{\defeq} TARGET\_HEIGHT \.{-} TRUSTED\_HEIGHT
 \.{+} 1 \.{\IN}}
\@x{\@s{16.4} state \.{\neq}\@w{working} \.{\implies}}
\@x{\@s{32.8} ( 2 \.{*} nprobes \.{\leq} N \.{*} ( N \.{-} 1 ) )}
\@pvspace{8.0pt}
\begin{lcom}{2.5}
\begin{cpar}{0}{F}{F}{0}{0}{}
We omit termination, as the algorithm deadlocks in the end.
 So termination can be demonstrated by finding a deadlock.
 Of course, one has to analyze the deadlocked state and see that
 the algorithm has indeed terminated there.
\end{cpar}
\end{lcom}
\@x{}\bottombar\@xx{}
\setboolean{shading}{false}
\begin{lcom}{0}
\begin{cpar}{0}{F}{F}{0}{0}{}
\ensuremath{\.{\,\backslash\,}\.{*}} Modification History
\end{cpar}
\begin{cpar}{0}{F}{F}{0}{0}{}
 \ensuremath{\.{\,\backslash\,}\.{*}} Last modified \ensuremath{Fri}
 \ensuremath{Jun} 26 12:08:28 \ensuremath{CEST} 2020 by \ensuremath{igor
}
\end{cpar}
\begin{cpar}{0}{F}{F}{0}{0}{}
 \ensuremath{\.{\,\backslash\,}\.{*}} Created \ensuremath{Wed}
 \ensuremath{Oct} 02 16:39:42 \ensuremath{CEST} 2019 by \ensuremath{igor
}
\end{cpar}
\end{lcom}

\end{document}